\documentclass[acmsmall]{acmart}

\AtBeginDocument{%
  }

\setcopyright{acmcopyright}
\copyrightyear{2018}
\acmYear{2018}
\acmDOI{XXXXXXX.XXXXXXX}

\acmConference[Conference acronym 'XX]{Make sure to enter the correct conference title from your rights confirmation emai}{June 03--05, 2018}{Woodstock, NY}
\acmPrice{15.00}
\acmISBN{978-1-4503-XXXX-X/18/06}

\usepackage{amsmath,amsfonts}
\usepackage{mathrsfs}
\usepackage{bm}
\usepackage{nicefrac}
\usepackage{booktabs}
\usepackage{array}
\usepackage{multirow}
\usepackage{threeparttable}
\usepackage{makecell}
\usepackage[procnumbered,ruled,vlined,linesnumbered]{algorithm2e}
\usepackage{siunitx}
\usepackage{stfloats}
\usepackage{subfigure}
\usepackage{hyperref}
\usepackage{array}
\usepackage{balance}
\usepackage{tabularx}
\usepackage{stfloats}
\usepackage{lipsum}
\usepackage{tikz,tikz-cd}
\usepackage{pgfplots,capt-of}

\pgfplotsset{compat=1.18}
\usetikzlibrary{plotmarks}

\newtheorem{problem}{Problem}
\newtheorem{theorem}{Theorem}[section]

\newtheorem{lemma}[theorem]{Lemma}
\newtheorem{definition}[theorem]{Definition}

\newcommand{\defeq}{\stackrel{\mathrm{def}}{=}}
\newcommand{\setof}[1]{\left\{#1 \right\}}
\newcommand{\abs}[1]{\left|#1 \right|}
\newcommand{\mean}[1]{\mathbb{E}\left[#1\right]}
\newcommand{\ceil}[1]{\left\lceil#1\right\rceil}
\newcommand{\Abs}[1]{\left\Vert#1\right\Vert}
\newcommand{\trace}[1]{\mathrm{Tr}\left(#1\right)}
\newcommand{\manc}[1]{H(#1)}

\newcommand{\mypar}[1]{\left( #1 \right)}
\newcommand{\bigpar}[1]{\big( #1 \big)}
\newcommand{\inv}[1]{{#1}^{-1}}
\newcommand{\pinv}[1]{\inv{\mypar{#1}}}
\newcommand{\indfunc}[1]{\chi\mypar{#1}}

\newcommand{\bsym}[1]{\boldsymbol{#1}}
\newcommand{\myord}[1]{{#1}^{\rm{th}}}
\newcommand{\rea}{\mathbb{R}}
\newcommand{\gr}{\mathcal{G}}
\newcommand{\veca}{\bsym{a}}
\newcommand{\vecb}{\bsym{b}}
\newcommand{\vecc}{\bsym{c}}
\newcommand{\vecd}{\bsym{d}}
\newcommand{\vece}{\bsym{e}}
\newcommand{\vecq}{\bsym{q}}
\newcommand{\vecone}{\bsym{1}}

\newcommand{\vecv}{\bsym{v}}
\newcommand{\vecpi}{\bsym{\pi}}
\newcommand{\vecpsi}{\bsym{\psi}}

\newcommand{\vecx}{\bsym{x}}
\newcommand{\vecy}{\bsym{y}}
\newcommand{\vecz}{\bsym{z}}

\newcommand{\pscr}{\mathscr{P}}

\newcommand{\kem}{\mathcal{K}}
\newcommand{\approxkem}{\tilde{\mathcal{K}}}
\newcommand{\subkem}[1]{\mathcal{K}_{#1}}

\newcommand{\lap}{\bsym{L}}
\newcommand{\mata}{\bsym{A}}
\newcommand{\matb}{\bsym{B}}
\newcommand{\matc}{\bsym{C}}
\newcommand{\matd}{\bsym{D}}
\newcommand{\matf}{\bsym{F}}
\newcommand{\matfstar}{\bsym{F}^*}
\newcommand{\mati}{\bsym{I}}
\newcommand{\matj}{\bsym{J}}

\newcommand{\matp}{\bsym{P}}
\newcommand{\matpdistr}{\bsym{P}'}
\newcommand{\matq}{\bsym{Q}}
\newcommand{\matr}{\bsym{R}}

\newcommand{\matw}{\bsym{W}}
\newcommand{\matx}{\bsym{X}}
\newcommand{\maty}{\bsym{Y}}
\newcommand{\matz}{\bsym{Z}}
\newcommand{\matpi}{\bsym{\Pi}}

\newcommand{\tilmatz}{\widetilde{\matz}}
\newcommand{\Otil}{\widetilde{O}}

\newcommand{\lemref}[1]{Lemma~\ref{#1}}
\newcommand{\thmref}[1]{Theorem~\ref{#1}}

\newcommand{\algoref}[1]{Algorithm~\ref{#1}}
\newcommand{\defref}[1]{Definition~\ref{#1}}
\newcommand{\secref}[1]{Section~\ref{#1}}
\newcommand{\tabref}[1]{Table~\ref{#1}}
\newcommand{\figref}[1]{Fig.~\ref{#1}}

\newcommand{\edge}[2]{\mypar{#1,#2}}

\DeclareMathOperator*{\argmin}{arg\,min}
\DeclareMathOperator*{\argmax}{arg\,max}

\DontPrintSemicolon
\SetKw{KwAnd}{and}
\SetFuncSty{textsc}
\SetKwInOut{Input}{Input\ \ \ \ }

\SetKwInOut{Output}{Output}

\begin{document}

\title[Means of Hitting Times for Random Walks on Graphs]{Means of Hitting Times for Random Walks on Graphs: Connections, Computation, and Optimization}

\author{Haisong Xia}
\email{hsxia22@m.fudan.edu.cn}
\affiliation{
    \institution{Fudan University}
    \city{Shanghai}
    \country{China}
}

\author{Wanyue Xu}
\email{xuwy@fudan.edu.cn}
\affiliation{
    \institution{Fudan University}
    \city{Shanghai}
    \country{China}
}

\author{Zuobai Zhang}
\email{17300240035@fudan.edu.cn}
\affiliation{
    \institution{Fudan University}
    \city{Shanghai}
    \country{China}
}

\author{Zhongzhi Zhang}
\email{zhangzz@fudan.edu.cn}
\affiliation{
    \institution{Fudan University}
    \city{Shanghai}
    \country{China}
}

\thanks{
    Haisong Xia, Wanyue Xu,  Zuobai~Zhang, and Zhongzhi Zhang are with the Shanghai Key Laboratory of Intelligent Information Processing, School of Computer Science, Fudan University, Shanghai 200433, China; Zhongzhi Zhang is also with the Shanghai Engineering Research Institute of Blockchains, Fudan University, Shanghai 200433, China; and Research Institute of Intelligent Complex Systems, Fudan University, Shanghai 200433, China. zhangzz@fudan.edu.cn
}
\thanks{
    An earlier version of this paper was presented in part at the Thirteenth ACM International Conference on Web Search and Data Mining (WSDM)~\cite{ZhXuZh20} [DOI: 10.1145/3336191.3371777].
}

\renewcommand{\shortauthors}{Xia et al.}

\begin{abstract}
    For random walks on graph $\mathcal{G}$ with $n$ vertices and $m$ edges, the mean hitting time $H_j$ from a vertex chosen from the stationary distribution to vertex $j$ measures the importance for $j$, while the Kemeny constant $\mathcal{K}$ is the mean hitting time from one vertex to another selected randomly according to the stationary distribution. In this paper, we first establish a connection between the two quantities, representing $\mathcal{K}$ in terms of $H_j$ for all vertices. We then develop an efficient algorithm estimating $H_j$ for all vertices and \(\mathcal{K}\) in nearly linear time of $m$. Moreover, we extend the centrality $H_j$ of a single vertex to $H(S)$ of a vertex set $S$, and establish a link between $H(S)$ and some other quantities. We further study the NP-hard problem of selecting a group $S$ of  $k\ll n$ vertices with minimum $H(S)$, whose objective function is monotonic and supermodular. We finally propose two greedy algorithms approximately solving the problem. The former has an approximation factor $(1-\frac{k}{k-1}\frac{1}{e})$ and $O(kn^3)$ running time, while the latter returns a $(1-\frac{k}{k-1}\frac{1}{e}-\epsilon)$-approximation solution in nearly-linear time of $m$, for any parameter $0<\epsilon <1$. Extensive experiment results validate the performance of our algorithms.
\end{abstract}



\maketitle

\section{Introduction}

As a powerful tool and method, random walks have found broad applications in various aspects. Frequently cited examples include image segmentation~\cite{Le06}, randomized algorithm design~\cite{SaDi12}, collaborative recommendation~\cite{FoPiReSa07}, community detection~\cite{LaDeBa14}, among others. A fundamental quantity related to random walks is hitting time~\cite{Lo93,ShZh19}, also called first-passage time~\cite{CoBeTeVoKl07}. For a random walk on a graph, the hitting time \(H_{ij}\) from a vertex \(i\) to another vertex \(j\) is the expected time for the walker starting from \(i\) to visit \(j\) for the first time. Hitting time is related to many problems and has been successfully applied to diverse areas, such as Hanoi problem with random move~\cite{WuZhCh11,QiDoZhZh20}, query suggestion~\cite{MeZhCh08}, and clustering algorithm~\cite{ChLiTa08}.

Except for the intrinsic interest of hitting time itself and its direct applications, many other relevant quantities related to random walks are encoded in or expressed in terms of this crucial quantity, for example, absorbing random-walk centrality~\cite{MaMaGi15} (or Markov centrality~\cite{WhSm03}), Kemeny constant~\cite{Hu14}, cover time~\cite{BlJaRy22,XuZh23}, and random detour time~\cite{BoRaZh11,RaZh13}. As the name implies, the absorbing random-walk centrality is a measure for the importance of vertices on a graph. For a vertex \(j\), its absorbing random-walk centrality \(H_j\) is defined by \(H_j=\sum_{i} \rho(i) H_{ij}\), where \(\rho(\cdot)\) is the starting probability distribution over all vertices. In the context of economic networks, each sector is represented by a vertex, while the flow of economic activity between sectors is modeled as weighted edges. To identify key sectors in a national economy, Bl\"ochl~\textit{et al.}~\cite{BlThVeFi11} selects the vertex with minimum absorbing random-walk centrality. This approach leverages absorbing random-walk centrality to reveal significant structural aspects of diverse national economies, which is a crucial task in data mining. Different from the shortest-path based centrality measures, random-walk based centrality metrics include the contributions from essentially all paths~\cite{Ne05}, and thus have a better discriminating power.

For random walks on a graph with \(n\) vertices, the Kemeny constant \(\kem\) is defined as the expected time for the walker starting from one vertex to second vertex selected randomly from the graph according to the stationary distribution \(\vecpi=\left(\pi_1, \pi_2, \cdots, \pi_n\right)^{\top}\) of the random walk, that is, \(\kem=\sum_{j} \pi_j H_{ij}\). The Kemeny constant has also found a wealth of applications in different fields~\cite{XiZh24SIGMOD}. It can be utilized to gauge the efficiency of user navigation through the World Wide Web (WWW)~\cite{LeLo02}. Moreover, the Kemeny constant is related to the mixing rate of an irreducible Markov chain~\cite{LePeWi09}, by regarding it as the expected time to mixing of the Markov chain~\cite{Hu06}. Recently, the Kemeny constant has been applied to measure the efficiency of robotic surveillance in network environments~\cite{PaAgBu15} and to characterize the noise robustness of a class of protocols for formation control~\cite{JaOl19}.

\textbf{Motivations}.
Despite the wide range of applications of the absorbing random-walk centrality and the Kemeny constant, it is a computational challenge to obtain their exact values. By definition, both the absorbing random-walk centrality and the Kemeny constant are a partial average of some hitting times.
However, the exact value of hitting time between any pair of vertices in a graph involves all eigenvalues and eigenvectors of (normalized) Laplacian matrix associated with the graph~\cite{Lo93,LiZh13PRE}, the time complexity of which is cubic on the vertex number. Thus, for large realistic networks with millions of vertices, we cannot obtain their absorbing random-walk centrality and the Kemeny constant by resorting this straightforward method for computing hitting time. It is then of theoretical and practical interest to seek for alternative approximate approaches that scale to large networks.

Moreover, in various practical applications, many problems cannot be reduced to determining the importance of individual nodes, but are essentially to find a group of \(k\) nodes that is the most important among all node groups, each containing exactly \(k\) nodes. For example, how to place resources on a fixed number of \(k\) peers in P2P networks, so that they are easily accessed by others~\cite{GkMiSa06}. Again for instance, in the field of wireless networks, sensor placement involves selecting an optimal subset of vertices to place sensors to sample physical signals~\cite{KrSiGu08,RaChVe14}, such as radiation or temperature. Finally, in the context of point cloud sampling~\cite{DiChWaBa20,ChTiFeVeKo17}, it requires selecting a representative subset of points to preserve the geometric features of reconstruction. These problems reveal an inherent connection between crucial node identification and data mining.

To solve the problem of identifying crucial nodes, the hitting time of node groups is utilized in many previous studies. For instance, the group hitting time can be regarded as the transmission cost of sensor groups in a wireless sensor network~\cite{MoBaZhPe20}. Under the capacity constraint \(k\), the minimization of transmission cost is transformed into minimizing the group hitting time. In addition, the group hitting time can be interpreted as the purchase probability of items in a recommendation system~\cite{YiCuLiYaCh12,JoNg17}, and the recommendation problem naturally leads to the minimization of group hitting time. Recently, the absorbing random-walk centrality for an individual vertex~\cite{WhSm03} has been extended to a group of vertices~\cite{MaMaGi15}. Moreover, the problem of choosing \(k\) most important vertices was proposed, and an approximated algorithm was developed to solve the problem in cubic running time, which is not applicable to large networks. Thus, it is of great significance to propose a hitting time based centrality for vertex groups and design fast algorithm to find the most important \(k\) vertices. For instance, this fast algorithm can be analogously applied in~\cite{BlThVeFi11} to identify crucial sector groups in a national economy, potentially informing policy decisions and economic analyses.

\textbf{Contributions}.
In a preliminary version~\cite{ZhXuZh20} of this paper, we studied the absorbing random-walk centrality $H_j=\sum_{i} \rho(i) H_{ij}$ for an individual vertex and Kemeny constant $\kem$. We focused on a special absorbing random-walk centrality $H_j=\sum_{i} \pi_i H_{ij}$, where the starting probability distribution $\rho(\cdot)$ is the stationary distribution $\vecpi=(\pi_1, \pi_2, \cdots, \pi_n)^{\top}$ of the random walk, and we call $H_j$ walk centrality. The walk centrality $H_j=\sum_{i} \pi_i H_{ij}$ has received considerable attention~\cite{TeBeVo09,Be09,Be16}. We first expressed \(\kem\) in terms of \(H_j\) for all vertices, and further expressed \(H_j\) and \(\kem\) in terms of quadratic forms of pseudoinverse of the Laplacian matrix. We then proposed a fast algorithm \textsc{ApproxHK} to compute approximate \(H_j\) for all vertices in $V$ and \(\kem\) for the whole graph in nearly linear time on the number of edges, based on the Johnson-Lindenstrauss lemma~\cite{Ac01} and the Laplacian solver~\cite{SpTe04,Sp10,KoMiPe11,LiBr12,CoKyMiPaPeRaSu14,KySa16}. Finally, we experimentally demonstrated that \textsc{ApproxHK} is accurate and is significantly faster than the direct exact computation of related quantities according to their definitions. The results from this preliminary version are primarily presented in Sections~\ref{sec:conn-walk-kem}, \ref{sec:approx-algo1}, and~\ref{subsec:perf-approxhk}. The key contributions of this preliminary version are summarized as follows.

\begin{itemize}
    \item We express the absorbing random-walk centrality \(H_j\) and the Kemeny constant \(\kem\) in terms of quadratic forms of pseudoinverse of the Laplacian matrix.
    \item We propose an approximation algorithm called \textsc{ApproxHK} to compute \(H_j\) and \(\kem\) in nearly linear time on the number of graph edges.
    \item Numerical experiments reveal the accuracy and efficiency of \textsc{ApproxHK}.
\end{itemize}

In the preliminary work~\cite{ZhXuZh20}, we focus solely on the mean hitting time for individual vertices. Our current study extends the walk centrality $H_j$ of individual vertices by proposing group walk centrality (GWC) for vertex groups. In a graph \(\gr=(V,E,w)\) with $n=|V|$ vertices and $m=|E|$ edges, the GWC $H(S)$ of a group $S \subseteq V$ is defined as $H(S) =\sum_{i \in V \setminus S} \pi_i H_{iS}$, where $H_{iS}$ is the expected time taken by a random walker starting from vertex $i$ to visit any vertex in \(S\) for the first time. We establish a connection between $H(S)$ and two other quantities related to hitting times.  We show that $H(\cdot)$ is a monotonic and supermodular set function. We then consider an NP-hard optimization problem MinGWC: how to find an optimal set $S^*$ of $k \ll n$ vertices such that $H(S^*)$ is minimized.  We devise two greedy algorithms,  \textsc{DeterMinGWC} and \textsc{ApproxMinGWC} to approximately solve this problem. The former has a $(1-\frac{k}{k-1} \frac{1}{e})$  approximation factor and $O(n^3)$ running time, while the latter has a $(1-\frac{k}{k-1} \frac{1}{e}-\epsilon)$ approximation factor for any error parameter \(\epsilon\in(0,1)\) and $\Otil (km \epsilon^{-2})$ running time, where $\Otil (\cdot)$ notation hides the ${\rm poly}\log $ factors. Finally, extensive experiment results demonstrate that both \textsc{DeterMinGWC} and \textsc{ApproxMinGWC} are accurate, while \textsc{ApproxMinGWC} is more efficient and is scalable to massive networks with millions of vertices. The results from our current presentation are detailed predominantly in Sections~\ref{sec:gwc}, \ref{sec:mingwc}, \ref{sec:approx-algo2}, and~\ref{subsec:perf-approxmingwc}. The key contributions are summarized as follows.

\begin{itemize}
    \item We propose a vertex group centrality called Group Walk Centrality (GWC) based on the group hitting time.
    \item We introduce the GWC minimization problem, proving its NP-hardness.
    \item Due to the monotonicity and supermodularity of GWC, we develop an efficient greedy algorithm called \textsc{ApproxMinGWC} based on \textsc{ApproxHK} to approximately solve the GWC minimization problem.
    \item Extensive experiment results demonstrate that both \textsc{ApproxHK} and \textsc{ApproxMinGWC} are accurate and scalable, which can be applied to networks with millions of vertices.
\end{itemize}

\section{Preliminaries}

In this section, we give a brief introduction to some notations, as well as some basic concepts about graphs, such as Laplacian matrix, resistance distance, random walks, hitting times, and some quantities derived from hitting times.

\subsection{Notations}

We use \(\rea\) to denote the set of real number, and we use normal lowercase letters such as \(a,b,c\) to represent scalars in \(\rea\). We denote vectors with bold lowercase letters such as \(\veca,\vecb,\vecc\), and matrices with bold uppercase letters like \(\mata,\matb,\matc\). For the convenience of representing specific elements in vectors and matrices, we use \(a_{i}\) to represent the \(\myord{i}\) element of vector \(\veca\), and use \(\mata_{[i,j]}\) to represent the entry at position \((i,j)\) in matrix \(\mata\). We use \(\mata_{[i,:]}\) and \(\mata_{[:,j]}\) to denote, respectively, the \(\myord{i}\) row and \(\myord{j}\) column of matrix \(\mata\). We write sets in subscripts to denote subvectors and submatrices. For example, \(\mata_{[S,j]}\) represents the submatrix of \(\mata\), which includes those matrix elements with row indices in \(S\) and column index as \(j\). Similarly, \(\veca_{-S}\) represents the subvector of \(\veca\) obtained from  \(\veca\) by removing elements with indices in set \(S\), and \(\mata_{-S}\) represents the submatrix of \(\mata\) obtained  from  \(\veca\) by removing elements with row indices and column indices in set \(S\). Note that the subscript takes precedence over the superscript, thus \(\mata_{-S}^{-1}\) denotes the inverse of \(\mata_{-S}\) rather than the submatrix of \(\mata^{-1}\). We use  \(\vecone_n\in \rea^n\) to denote a vector of \(n\) dimensions with all elements being \(1\). Sometimes we skip subscripts if there is no ambiguity. For a matrix \(\mata\in\rea^{m\times n}\), its Frobenius norm is \(\Abs{\mata}_F=\sqrt{\trace{\mata^\top\mata}}\).

\begin{definition}[\(\epsilon\)-approximation]
    Let \(x\) and \(\tilde{x}\) be positive scalars, with \(\epsilon\) as the error parameter such that \(\epsilon\in(0,1)\). We refer to \(\tilde{x}\) as an \(\epsilon\)-approximation of \(x\) if the inequality \((1-\epsilon)x\le \tilde{x}\le(1+\epsilon)x\) holds. For convenience, we write this as \(x\approx_{\epsilon}\tilde{x}\).
\end{definition}

\subsection{Graph and Laplacian Matrix}\label{sub:lap}

Let \(\gr=(V,E,w)\) denote a connected undirected weighted graph or network,  where \(V\) is the set of vertices,  \(E\) is the set of edges, and  \(w: E\to \mathbb{R}_{+}\) is the positive edge weight function, with \(w_e\) being the weight for edge \(e\). Then, there are total \(n=\abs{V}\) vertices and \(m=|E|\) edges in graph \(\gr\). Let \(w_{\max}\) and \(w_{\min}\) denote the maximum edge weight and minimum edge weight, respectively. Namely, \(w_{\max}=\max_{e\in E} w_e \) and \(w_{\min}=\min_{e\in E} w_e\).

Mathematically, the topological and weighted properties of a graph \(\gr\) are encoded in its generalized adjacency matrix \(\mata\) with the entry \(a_{ij}\) denoting the adjacency relation between vertices \(i\) and \(j\). If vertices \(i\) and \(j\) are linked to each other by an edge, then we denote this edge as \(e=\edge{i}{j}\), with \(a_{ij}= a_{ji}=w_{e}>0\). If vertices \(i\) and \(j\) are not adjacent, \(a_{ij}=a_{ji}=0\). In a weighted graph \(\gr\), the degree \(d_i\) of a vertex \(i\) is defined by \(d_i=\sum_{j=1}^n a_{ij}\)~\cite{BaBaPaVe04}, where we denote the maximal vertex degree as \(d_{\max}=\max\setof{d_i|i\in V}\).
The diagonal degree vector of graph \(\gr\) is defined to be \(\vecd=\mypar{d_1,d_2,\cdots,d_n}^\top\), while the diagonal degree matrix of graph \(\gr\) is defined to be \({\matd}={\rm diag}(d_1, d_2,\ldots,d_n)\), and the Laplacian matrix of \(\gr\) is \({\lap}={\matd}-{\mata}\). Thus, \(\lap\)  is a symmetric diagonally dominant (SDD) matrix. Moreover, for
any  nonempty  proper subset  \(S \) of vertex set $V$, \(\lap_{-S}\) is a symmetric diagonally dominant M-matrix (SDDM).

Let \(\matb \in \mathbb{R}^{|E| \times \abs{V}}\) be the incidence matrix of graph \(\gr\). For each edge \(e\) with two end vertices \(i\) and \(j\), a direction is assigned arbitrarily. Let \(\vecb_e^\top\) be the row of matrix \(\matb\) associated with edge \(e\). Then the element \(b_{eu}\) at row corresponding to edge \(e\) and column corresponding to vertex \(u\) is defined as follows: \(b_{eu} = 1\) if vertex \(u\) is the tail of edge \(e\), \(b_{eu}=-1\) if vertex \(u\) is the head of  edge \(e\), and \(b_{eu}=0\) otherwise. Let \(\vece_u\) be the \(u\)-th canonical basis of the space \(\mathbb{R}^{\abs{V}}\), then for an edge \(e\) connecting two vertices \(i\) and \(j\), \(\vecb_e\) can also be recast as \(\vecb_e=\vece_{i}-\vece_{j}\).  Let  \(\matw \in \mathbb{R}^{|E| \times |E|}\) be a diagonal matrix with every diagonal entry corresponding to the weight  \(w_e\) of an edge  $e\in E$. Then the Laplacian matrix \(\lap\) of graph \(\gr\) can be written as \(\lap=\matb^T\matw\matb=\sum_{e\in E}w_e\vecb_e\vecb_e^{\top}\).

The Laplacian matrix \(\lap\) is symmetric and positive semidefinite. All  its eigenvalues  are non-negative, with a unique zero eigenvalue. Let \(0=\lambda_1< \lambda_2 \leq \lambda_3\leq \dots\leq \lambda_{n}\) be the \(n\) eigenvalues of  \(\lap\), and let \(u_i\), \( i={1,2,\dots,n}\), be their corresponding mutually orthogonal  unit eigenvectors. Then, \(\lap\) has the following spectral decomposition:  \(\lap=\sum_{i=2}^{n}\lambda_i u_iu_i^\top\).  It is easy to verify that \( \lambda_{n}\leq n w_{\max}\)~\cite{LiSc18}.
Since \(\lap\) is not invertible, we use \(\lap^{\dagger}\) to denote its pseudoinverse, which can be written as \(\lap^{\dagger}=\sum_{i=2}^{n}\frac{1}{\lambda_i}u_iu_i^{\top}\). Let \(\matj\) denote the matrix with all entries being ones. Then the pseudoinverse \(\lap^{\dagger}\) can also be recast as \(\mypar{\lap +\frac{1}{n}\matj}^{-1} - \frac{1}{n}\matj\)~\cite{GhBoSa08}. Note that for a general symmetric matrix, it shares the same null space as its Moore-Penrose generalized inverse~\cite{BeGrTh74}.
Since the null space of null of \({\lap}\) is \( \vecone\), it turns out that \({\lap}\vecone={\lap}^{\dagger}\vecone=\mathbf{0}\). Note that for an invertible matrix, its pseudoinverse is exactly its inverse.

\subsection{Electrical Network and Resistance Distance}

For an arbitrary graph \(\gr=(V,E,w)\), we can define its corresponding electrical network \(\bar{\gr}=(V,E,r)\), which is obtained from \(\gr\)  by considering edges as resistors and considering vertices as junctions between resistors~\cite{DoSn84}. The resistor of an associated edge \(e\) is \(r_e=w_e^{-1}\). For graph \(\gr\), the resistance distance \(R_{ij}\) between two vertices \(i\) and \(j\)  is defined as the effective resistance between \(i\) and \(j\) in the corresponding  electrical network \(\bar{\gr}\)~\cite{KlRa93}, which is equal to the potential difference between \(i\) and \(j\) when a unit current enters one vertex and leaves the other one.

For graph \(\gr\), the resistance distance \(R_{ij}\) between two vertices \(i\) and \(j\) can be expressed in terms of the elements of \(\lap^{\dagger}\) as~\cite{KlRa93}:
\begin{equation}\label{EE04}
    R_{ij}={\lap}_{[i,i]}^{\dagger}+{\lap}_{[j,j]}^{\dagger}-2{\lap}_{[i,j]}^{\dagger}.
\end{equation}
Define \(\matr\) as the \(n \times n\) resistance matrix of graph \(\gr\), whose entry \(R_{[i,j]}\) at row \(i\) and column \(j\) represents the resistance distance \(R_{ij}\) between vertices \(i\) and \(j\).

\begin{lemma}\label{Foster}\cite{Te91}
    Let \(\gr=(V,E,w)\) be a simple connected graph with \(n\) vertices. Then the sum of  weight times resistance distance over all pairs of adjacent vertices in  \(\gr\)  satisfies
    \begin{equation*}
        \sum_{ \edge{i}{j}\in E }w_{ij}R_{ij}=n-1.
    \end{equation*}
\end{lemma}

Similarly, We can also define the resistance distance \(R_{iS}\) between a vertex $i$ and a group  \(S\) of vertices. In the electrical network \(\bar{\gr}\) corresponding to graph \(\gr\), if all  the vertices in group \(S\) are grounded, the voltage of all vertices in set  \(S\) is always zero.
Then the resistance distance \(R_{iS}\) between vertex \(i\) and vertex group \(S\) is defined as the voltage of \(i\) when a unit current enters \(\bar{\gr}\) at vertex $i$ and leaves \(\bar{\gr}\) at  vertices in \(S\).  Let \(\pscr_{iS}\) denote a simple path connecting vertex \(i\) and an arbitrary vertex in \(S\). Then, the following relation holds:
\begin{equation}\label{eq:phy-resist}
    R_{iS}\le \sum_{\edge{a}{b}\in\pscr_{iS}}w_{ab}^{-1}\,.
\end{equation}

\subsection{Random Walk on a Graph}

For a connected weighted graph \(\gr\) with \(n\) vertices, the classical random walk on \(\gr\) can be described by the transition matrix \(\matp\in\rea^{n\times n}\). At any time step, the walker located at vertex \(i\) moves to vertex \(j\) with probability \({\matp_{[i,j]}=d_i^{-1}\mata_{[i,j]}}\).
It is easy to verify that
\begin{equation}\label{eq:trs}
    \matp=\matd^{-1}\mata.
\end{equation}

If  \(\gr\) is finite and non-bipartite, the random walk has a unique stationary distribution~\cite{LiZh13PRE}
\begin{equation}\label{EE01}
    \vecpi=(\pi_1, \pi_2, \cdots, \pi_n)^{\top}=\left(\frac{d_1}{d}, \frac{d_2}{d}, \cdots, \frac{d_n}{d}\right)^{\top},
\end{equation}
where \(d\) is the sum of degrees over all vertices, namely \(d=\sum_{i=1}^n d_i=\sum_{i=1}^{n}\sum_{j=1}^{n} a_{ij}\).

A fundamental quantity for random walks is hitting time~\cite{Lo93,CoBeTeVoKl07}. The hitting time \(H_{ij}\) from vertex \(i\) to vertex \(j\),  is the expected number of jumps for a walker starting  from \(i\) to visit \(j\) for the first time.
In other words, if we denote the time steps for a walker starting from \(i\) to first reach \(j\) as the random variable \(T_{ij}\), then we have \(H_{ij}=\mean{T_{ij}}\).
There is an intimate relationship between hitting time and resistance distance~\cite{Te91}.
\begin{lemma}
    Let \(\gr\) be a connected weighted graph with  resistance matrix  \(\matr\). Let \(H_{ij}\) be the hitting time  from vertex \(i\) to vertex \(j\). Then,
    \begin{equation}\label{EE03}
        H_{ij}=\frac{1}{2}\sum_{z=1}^{n} d_z(R_{ij}+R_{jz}-R_{iz}).
    \end{equation}
\end{lemma}

Moreover, hitting time is also related to the fundamental matrix under the non-absorbing random walk model.

\begin{lemma}[\cite{BoRaZh11}]\label{lem:non-absorb-fund}
    For the non-absorbing random walk model, let \(\matfstar=(\mati-\matp+\vecone\vecpi^\top)^{-1}\matpi^{-1}\) denote the fundamental matrix, where \(\matpi\) is defined as \(\textrm{diag}(\vecpi)\). Then, the hitting time \(H_{ij}\) can be represented in terms of the elements of matrix \(\matfstar\) as
    \begin{equation*}
        H_{ij}=\matfstar_{[j,j]}-\matfstar_{[i,j]}.
    \end{equation*}
\end{lemma}

A lot of interesting quantities can be derived from hitting times. Here we only consider three quantities, the absorbing random-walk centrality~\cite{WhSm03,MaMaGi15}, the Kemeny constant~\cite{XiZh24KDD}, and the random detour time~\cite{BoRaZh11,RaZh13}.

For a vertex \(j\) in graph \(\gr=(V,E,w)\), its absorbing random-walk centrality \(H_j\) is defined as \(H_j=\sum_{i} \rho(i) H_{ij}\), where \(\rho(\cdot)\) is the starting probability distribution over all vertices in \(V\). By definition, \(H_j\) is a weighted average of hitting times to vertex \(j\). The smaller the value of \(H_j\), the more important the vertex \(j\). The random-walk based centrality has an obvious advantage over those shortest-path based centrality measures~\cite{Ne05}. In~\cite{WhSm03,MaMaGi15}, the uniform distribution for \(\rho(\cdot)\) is considered for the starting  vertices. In this paper, we concentrate on a natural choice of \(\rho(\cdot)\) by selecting the starting vertex from the stationary distribution \(\vecpi\). In this case, \(H_j=\sum_{i}\pi_i H_{ij}\), which has been much studied~\cite{TeBeVo09,Be09,Be16}. In the following text, we call \(H_j=\sum_{i} \pi_i H_{ij}\) \textit{walk centrality} for short.

Another quantity we are concerned with is the Kemeny constant \(\kem\). For a graph \(\gr\), its Kemeny constant \(\kem\) is defined as the expected steps for a walker starting from vertex \(i\) to vertex \(j\) selected randomly from the vertex set \(V\), according to the stationary distribution \(\vecpi\). That is, \(\kem = \sum_{j = 1}^{n} \pi_j H_{ij}\). The Kemeny constant has been used to measure the user navigation efficiency through the WWW~\cite{LeLo02} and  robotic surveillance efficiency in network environments~\cite{PaAgBu15}. It can also measure the mixing rate of random walks~\cite{LePeWi09}.

Most quantities for random walks on graph \(\gr\) are determined by the eigenvalues and eigenvectors of the normalized Laplacian matrix~\cite{Ch97}, \({\matd}^{-\frac{1}{2}}\lap {\matd}^{-\frac{1}{2}}\), of \(\gr\), including the walk centrality and Kemeny constant. By definition, \({\matd}^{-\frac{1}{2}}\lap {\matd}^{-\frac{1}{2}}\) is a real, symmetric, semi-definitive matrix. Let \(0=\sigma_1 < \sigma_2 \leq \sigma_3 \leq \cdots \leq \sigma_n \) be the \(n\) eigenvalues of the normalized Laplacian matrix \({\matd}^{-\frac{1}{2}}\lap {\matd}^{-\frac{1}{2}}\). And let \(\vecpsi_1\), \(\vecpsi_2\), \(\vecpsi_3\), \(\ldots\), \(\vecpsi_n\) be their corresponding mutually orthogonal eigenvectors of unit length, where \(\vecpsi_i=(\psi_{i1},\psi_{i2},\ldots,\psi_{in})^{\top}\). Then~\cite{Lo93,Be16},
\begin{equation}\label{ATT01}
    H_j=\sum_{i=1}^{n} \pi_i H_{ij}=\frac{d}{d_j}\sum_{k=2}^{n}\frac{1}{\sigma_{k}}\psi_{kj}^{2}
\end{equation}
and
\begin{equation}\label{Kemeny01}
    \kem =\sum_{j=1}^{n}\pi_j\,H_{ij} =\sum_{k=2}^{n}\frac{1}{\sigma_{k}}\,.
\end{equation}

Equations~\eqref{ATT01} and~\eqref{Kemeny01} provide exact computation for the walk centrality and Kemeny constant, respectively. However, both formulas are expressed in terms of the eigenvalues and eigenvectors of the normalized Laplacian, the computation complexity for which scale as \(O(n^3)\). Thus, direct  computation for \(H_j\) and \(\kem\)  using spectral method appears to be prohibitive for large networks, and is infeasible to those realistic networks with millions of vertices.

\section{Connections between Walk Centrality and Kemeny Constant}\label{sec:conn-walk-kem}

Although both the walk centrality \(H_j\) and the Kemeny constant \(\kem\) have attracted much attention from the scientific community, the relation between them for a general graph  \(\gr\) is still lacking.
In this section, we establish an explicit relation between the walk centrality \(H_j\) and the Kemeny constant \(\kem\).
We then express both quantities in terms of quadratic forms of the pseudoinverse \(\lap^{\dagger}\) of graph Laplacian \(\lap\), which is helpful to address their computational challenges.

First, we show that the Kemeny constant \(\kem\) can be expressed in a linear combination of the walk centrality \(H_j\) for all vertices in \(\gr\), as stated in the following lemma.

\begin{lemma}
    Let \(\gr=(V,E,w)\) be a connected weighted graph. Then, its Kemeny constant \(\kem\) and walk centrality \(H_j\) obey the following relation:
    \begin{equation}\label{HjK01}
        \kem=\sum_{j=1}^{n} \pi_j H_j=\sum_{j=1}^{n} \frac{d_j}{d} H_j.
    \end{equation}
\end{lemma}
\begin{proof}
    From~\eqref{Kemeny01}, the Kemeny constant \(\kem\) is independent of the starting vertex \(i\). Define \(\subkem{i} =\sum_{j=1}^{n}\pi_j\,H_{ij}\). Then   \(\subkem{i}=\subkem{j}\) holds for any pair of vertices \(i\) and \(j\). Thus, we have
    \begin{equation*}\label{Kemeny02}
        \kem =\subkem{i}=\sum_{i=1}^{n} \pi_i \left( \sum_{j=1}^{n}\pi_j\,H_{ij}\right)
        =\sum_{j=1}^{n} \pi_j \left( \sum_{i=1}^{n}\pi_i\,H_{ij}\right)
        =\sum_{j=1}^{n} \frac{d_j}{d} H_j,
    \end{equation*}
    which establishes the lemma.
\end{proof}

After obtaining the relation governing the Kemeny constant \(\kem\) and walk centrality \(H_j\), we continue to express them in terms of quadratic forms of matrix \(\lap^{\dagger}\).

\begin{lemma}\label{HjK}
    Let \(\gr=(V,E,w)\) be a connected weighted graph  with  Laplacian matrix \(\lap\). Then, the walk centrality \(H_j\) and Kemeny constant \(\kem\) can be represented  in terms of quadratic forms of the pseudoinverse \(\lap^{\dagger}\)  of   matrix  \(\lap\) as:
    \begin{equation}\label{ATT02}
        H_j=d(\vece_j - \vecpi)^{\top} \lap^{\dagger} (\vece_j - \vecpi)
    \end{equation}
    and
    \begin{equation}\label{Kemeny03}
        \kem= \sum_{j=1}^{n} d_j (\vece_j - \vecpi)^{\top} \lap^{\dagger} (\vece_j - \vecpi).
    \end{equation}
\end{lemma}
\begin{proof}
    We first prove~\eqref{ATT02}. Inserting~(\ref{EE03}) and~(\ref{EE04}) into~(\ref{ATT01}) leads to
    \begin{equation}\label{EE05}
        H_j =\frac{1}{2} \sum_{i=1}^{n} \pi_i \sum_{z=1}^{n} d_z(R_{ij}+R_{jz}-R_{iz})
        =\frac{1}{d} \sum_{i=1}^{n} d_i \sum_{z=1}^{n} d_z \left({\lap}_{[j,j]}^{\dagger}-{\lap}_{[i,j]}^{\dagger}-{\lap}_{[j,z]}^{\dagger}+{\lap}_{[i,z]}^{\dagger}\right)
    \end{equation}

    The four terms in the brackets of~\eqref{EE05} can be sequentially calculated as follows:
    \begin{equation}\label{EE06}
        \sum_{i=1}^{n} d_i \sum_{z=1}^{n} d_z {\lap}_{[j,j]}^{\dagger}=d^2 \vece_j^{\top} \lap^{\dagger} \vece_j\,,
    \end{equation}
    \begin{equation}\label{EE07}
        \sum_{i=1}^{n} d_i \sum_{z=1}^{n} d_z {\lap}_{[i,j]}^{\dagger}=\sum_{i=1}^{n} d_i \sum_{z=1}^{n} d_z {\lap}_{[j,z]}^{\dagger}=d^2 \vece_j^{\top} \lap^{\dagger} \vecpi\,,
    \end{equation}
    and
    \begin{equation}\label{EE08}
        \sum_{i=1}^{n} d_i  \sum_{z=1}^{n} d_z {\lap}_{[i,z]}^{\dagger}=d^2 \vecpi^{\top} \lap^{\dagger} \vecpi \,.
    \end{equation}
    Plugging~\eqref{EE06},~\eqref{EE07}, and~\eqref{EE08} into~\eqref{EE05}, we obtain
    \begin{equation}\label{EE09}
        H_j  =d(\vece_j^{\top} \lap^{\dagger} \vece_j - 2 \vece_j^{\top} \lap^{\dagger} \vecpi + \vecpi^{\top} \lap^{\dagger} \vecpi)
        =d(\vece_j - \vecpi)^{\top} \lap^{\dagger} (\vece_j - \vecpi).
    \end{equation}
    Substituting~\eqref{EE09}  into~\eqref{HjK01} gives~\eqref{Kemeny03}.
\end{proof}

\section{Fast Approximation Algorithm for Walk Centrality and Kemeny Constant}\label{sec:approx-algo1}

In the preceding section, we reduce the problem of computing \(H_j\) and \(\kem\) to evaluating the quadratic forms \((\vece_u - \vecpi)^{\top} \lap^{\dagger} (\vece_u- \vecpi)\), \(u=1,2,\ldots, n\), of matrix  \(\lap^{\dagger}\).
However,  this involves computing the pseudoinverse of \(\lap\), the straightforward computation for which still has a complexity of \(O(n^3)\), making it infeasible to huge networks. Here, we present an algorithm to compute an approximation of \(H_u\) for all \(u \in V\) and \(\kem\) in nearly linear time with respect to the number of edges, which has a strict theoretical guarantee with high probability.

Let \(C(u)=(\vece_u - \vecpi)^{\top} \lap^{\dagger} (\vece_u- \vecpi)\), which can be written in an Euclidian norm as
\begin{small}
    \begin{equation}\label{EE12}
        \begin{split}
            C(u) & =(\vece_u - \vecpi)^{\top} \lap^{\dagger} \lap \lap^{\dagger} (\vece_u - \vecpi)
            =(\vece_u - \vecpi)^{\top} \lap^{\dagger} \matb^{\top} \matw \matb \lap^{\dagger} (\vece_u - \vecpi)                                          \\
                 & =(\vece_u - \vecpi)^{\top} \lap^{\dagger} \matb^{\top} \matw^{\frac{1}{2}} \matw^{\frac{1}{2}} \matb \lap^{\dagger} (\vece_u - \vecpi)
            =\|\matw^{\frac{1}{2}} \matb \lap^{\dagger} (\vece_u - \vecpi)\|^2.
        \end{split}
    \end{equation}
\end{small}
This in fact equals the square of the distance between vectors \(\matw^{\frac{1}{2}} \matb \lap^{\dagger} \vece_u\) and \(\matw^{\frac{1}{2}} \matb \lap^{\dagger} \vecpi\), which can be evaluated by the Johnson-Lindenstraus
Lemma (JL Lemma)~\cite{Ac01}.

\begin{lemma}[JL Lemma~\cite{Ac01}]
    \label{lemma:JL}
    Given $n$ fixed vectors \(\vecv_1,\vecv_2,\ldots,\vecv_n\in \mathbb{R}^d\) and
    \(\epsilon>0\), let
    \(\matq_{k\times d}\),  \(k\ge 24\log n/\epsilon^2\), be a matrix with each  entry equal  to \(1/\sqrt{k}\) or \(- 1/\sqrt{k}\)  with the same probability \(1/2\). Then with probability at least \(1-1/n\),
    \begin{equation*}
        \Abs{\vecv_i-\vecv_j}^2\approx_\epsilon\Abs{\matq\vecv_i-\matq\vecv_j}^2
    \end{equation*}
    for all all pairs  \(i,j\le n\).
\end{lemma}

\lemref{lemma:JL} indicates that,  the pairwise distances \(\|\vecv_i-\vecv_j\|^2\) (\(i,j=1,2,\ldots, n\)) are almost preserved if we project the \(n\) vectors \(\vecv_i\) (\(i=1,2,\ldots, n\))  into a lower-dimensional space, spanned
by \(O(\log n)\) random vectors.

In order to compute  \(C(u)\), we use Lemma~\ref{lemma:JL} to reduce the dimension. Let \(\matq\) be a \(k\times m\) random projection matrix. Then  \(\|\matq \matw^{\frac{1}{2}} \matb \lap^{\dagger}(\vece_u-\vecpi)\|\) is a good approximation for \(\|\matw^{\frac{1}{2}} \matb \lap^{\dagger} (\vece_u - \vecpi)\|\). Here we can use sparse matrix multiplication to compute \(\matq \matw^{\frac{1}{2}} \matb\). However, computing \(\matz=\matq \matw^{\frac{1}{2}} \matb \lap^{\dagger}\) directly involves inverting \(\lap+\frac{1}{n}\matj\). We avoid this by solving the system of equations \(\lap \vecz_i=\vecq_i\), \(i=1,\ldots,k\), where  \(\vecz^\top_i\) and \(\vecq^\top_i\) are, respectively, the \(i\)-th row of \(\matz\) and \(\matq \matw^{\frac{1}{2}} \matb\). For the convenience  of description, in the sequel we use the notation \(\Otil(\cdot)\) to hide \(\mathrm{poly} \log \) factors. By using Laplacian solvers~\cite{SpTe04,Sp10,KoMiPe11,LiBr12,CoKyMiPaPeRaSu14,KySa16}, \(\vecz^\top _i\) can be efficiently approximated.  We here use the  solver from~\cite{CoKyMiPaPeRaSu14}, the performance of which is characterized in \lemref{lemma:ST}.

\begin{lemma}[SDD Solver~\cite{CoKyMiPaPeRaSu14,KySa16}]
    \label{lemma:ST}
    There is an algorithm \(\vecx = \mathtt{Solver}(\lap,\vecy,\delta)\) which
    takes an SDDM matrix or a Laplacian \(\lap\),
    a column vector \(\vecy\), and an error
    parameter \(\delta > 0\), and returns a column vector \(\vecx\) satisfying  \(\boldsymbol{1}^\top\vecx = 0\) and
    \[
        \|\vecx - \lap^{\dagger} \vecy\|_{\lap} \leq \delta \|\lap^{\dagger} \vecy\|_{\lap},
    \]
    where \(\Abs{\vecy}_{\lap} = \sqrt{\vecy^{\top} \lap \vecy}\).
    The algorithm runs in expected time \(\Otil \left(m \log(1/\delta) \right)\).
\end{lemma}

By Lemmas~\ref{lemma:JL} and~\ref{lemma:ST}, one can approximate \(C(u)\) arbitrarily well.

\begin{lemma}\label{lem:error1}
    Given an approximation factor \(\epsilon \le 1/2\) and a \(k\times n\) matrix \(\matz\) that satisfies
    \begin{equation*}
        C(u)\approx_\epsilon\Abs{\matz\mypar{\vece_u-\vecpi}}^2,
    \end{equation*}
    for any vertex \(u\in V\) and
    \begin{equation*}
        \Abs{\matw^{\frac{1}{2}}\matb\lap^\dagger\mypar{\vece_u-\vecv_v}}^2\approx_\epsilon\Abs{\matz\mypar{\vece_u-\vecv_v}}^2
    \end{equation*}
    for any pair of vertices \(u,v \in V\).
    Let \(\vecz_i\) be the \(i\)-th row of \(\matz\) and let \(\tilde{\vecz}_i\) be an approximation of \(\vecz_i\) for all \(i \in \{1,2,...,k\}\), satisfying
    \begin{equation}\label{EE13} \|\vecz_i-\tilde{\vecz}_i\|_{\lap}\le\delta
        \|\vecz_{i}\|_{\lap},
    \end{equation}
    where
    \begin{equation}\label{EE14}
        \delta \leq  \frac{\epsilon }{3} \frac{d-d_u}{d}
        \sqrt{\frac{(1-\epsilon) w_{\min}}{(1+\epsilon) n^4 w_{\max}}}.
    \end{equation}
    Then for any vertex \(u\) belonging to \(V\),
    \begin{align}
        \label{EE15}
        (1 - \epsilon)^2  C(u)
        \leq
        \|\tilmatz  (\vece_{u} - \vecpi)\|^2
        \leq
        (1 + \epsilon)^2  C(u),
    \end{align}
    where \(\tilmatz = [\tilde{\vecz}_1, \tilde{\vecz}_2, ..., \tilde{\vecz}_k]^\top\).
\end{lemma}
\begin{proof}
    To prove (\ref{EE15}), it is sufficient  to show that for an arbitrary vertex \(u\),
    \begin{equation}\label{EE16a}
        \begin{split}
            \abs{\Abs{\matz(\vece_u-\vecpi)}^2-\|{\tilmatz(\vece_u-\vecpi)}\|^2}
             & =     \abs{\Abs{\matz(\vece_u-\vecpi)}-\|\tilmatz(\vece_u-\vecpi)\|}\times\abs{\Abs{\matz(\vece_u-\vecpi)}+\|\tilmatz(\vece_u-\vecpi)\|} \\
             & \leq  \left(\frac{2\epsilon}{3}+\frac{\epsilon^2}{9}\right)\Abs{\matz(\vece_u-\vecpi)}^2,
        \end{split}
    \end{equation}
    which is obeyed if
    \begin{small}
        \begin{equation}\label{EE16}
            \abs{\Abs{\matz(\vece_u-\vecpi)}-\|\tilmatz(\vece_u-\vecpi)\|} \le
            \frac{\epsilon}{3}\Abs{\matz(\vece_u-\vecpi)}.
        \end{equation}
    \end{small}
    This can be understood from the following arguments. On the one hand, if
    {\footnotesize \(\abs{\Abs{\matz(\vece_u-\vecpi)}^2-\|{\tilmatz(\vece_u-\vecpi)}\|^2} \le \left(\frac{2\epsilon}{3}+\frac{\epsilon^2}{9}\right)\Abs{\matz(\vece_u-\vecpi)}^2\)}, then
    \begin{equation*}
        \left(1-\frac{2\epsilon}{3}-\frac{\epsilon^2}{9}\right)\Abs{\matz(\vece_u-\vecpi)}^2\leq\|{\tilmatz(\vece_u-\vecpi)}\|^2\leq\left(1+\frac{2\epsilon}{3}+\frac{\epsilon^2}{9}\right)\Abs{\matz(\vece_u-\vecpi)}^2,
    \end{equation*}
    which, combining with \(\epsilon \le 1/2\) and the assumption that \(C(u)\approx_\epsilon\Abs{\matz\mypar{\vece_u-\vecpi}}^2\), leads to Eq.~\eqref{EE15}.
    On the other hand, if Eq.~\eqref{EE16} is true, we have \(\|\tilmatz(\vece_u-\vecpi)\| \le (1+\frac{\epsilon}{3})\Abs{\matz(\vece_u-\vecpi)}\). Thus,
    \begin{equation*}
        \abs{\Abs{\matz(\vece_u-\vecpi)}+\|\tilmatz(\vece_u-\vecpi)\|} \le (2+\frac{\epsilon}{3})\Abs{\matz(\vece_u-\vecpi)},
    \end{equation*}
    which results in Eq.~\eqref{EE16a}.

    We now prove that~\eqref{EE16} holds true. By applying triangle inequality and Cauchy-Schwarz inequality, we obtain
    \begin{align*}
        \abs{\Abs{\matz(\vece_u-\vecpi)}-\|\tilmatz(\vece_u-\vecpi)\|}
         & \leq \Abs{(\matz-\tilmatz)(\vece_u-\vecpi)}=\frac{1}{d} \Abs{\sum_{v=1}^{n} d_v (\matz-\tilmatz)(\vece_u-\vece_v)}                                                              \\
         & \leq \frac{1}{d} \sum_{v=1}^{n} d_v \Abs{(\matz-\tilmatz)(\vece_u-\vece_v)}\leq\frac{1}{d} \sqrt{\sum_{v=1}^{n} d_v^2 \sum_{v=1}^{n} \Abs{(\matz-\tilmatz)(\vece_u-\vece_v)}^2} \\
         & \leq \sqrt{\sum_{v=1}^{n} \Abs{(\matz-\tilmatz)(\vece_u-\vece_v)}^2},
    \end{align*}
    where the last inequality is due to the fact that \(d=\sum_{v=1}^{n} d_v\ge \sqrt{\sum_{v=1}^{n} d_v^2}\).

    Since we only consider connected graphs,   there exists a simple path \(\mathcal{P}_v\) between \(u\) and any other vertex \(v\).  Applying the triangle inequality along path \(\mathcal{P}_v\), the square of the last sum term in the above equation can be evaluated as:
    \begin{align*}
        \sum_{v=1}^{n} \Abs{(\matz-\tilmatz)(\vece_u-\vece_v)}^2
         & \leq\sum_{v=1}^{n} \left(\sum_{\edge{a}{b} \in \mathcal{P}_v} \Abs{(\matz-\tilmatz)(\vece_a-\vece_b)}\right)^2
        \leq n \sum_{v=1}^{n} \sum_{\edge{a}{b} \in \mathcal{P}_v} \Abs{(\matz-\tilmatz)(\vece_a-\vece_b)}^2                                       \\
         & \leq n^2 \sum_{\edge{a}{b} \in E} \Abs{(\matz-\tilmatz)(\vece_a-\vece_b)}^2=n^2 \Abs{(\matz - \tilmatz ) \matb^{\top}}_{F}^{2}          \\
         & =n^2 \Abs{\matb(\matz - \tilmatz )^{\top}}_{F}^{2}\leq\frac{n^2}{w_{\min}} \Abs{\matw^{1/2} \matb(\matz - \tilmatz )^{\top}}_{F}^{2}\,.
    \end{align*}
    The last term can be bounded as
    \begin{align*}
        \frac{n^2}{w_{\min}} \Abs{\matw^{1/2} \matb(\matz - \tilmatz )^{\top}}_{F}^{2}
         & =\frac{n^2}{w_{\min}} \trace{(\matz - \tilmatz )\matb^{\top} \matw \matb(\matz - \tilmatz )^{\top}} \\
         & =\frac{n^2}{w_{\min}} \trace{(\matz - \tilmatz )\lap(\matz - \tilmatz )^{\top}}
        =\frac{n^2}{w_{\min}} \sum_{i=1}^k (\vecz_i - \tilde{\vecz}_i)^\top \lap (\vecz_i - \tilde{\vecz}_i)   \\
         & \leq\frac{n^2\delta^2}{w_{\min}} \sum_{i=1}^k \vecz_i^\top\lap\vecz_i
        =\frac{n^2\delta^2}{w_{\min}}  \trace{\matz\lap\matz^{\top}}
        =\frac{n^2\delta^{2}}{w_{\min}} \Abs{\matw^{1/2} \matb \matz^{\top}}_{F}^{2}\,,
    \end{align*}
    where the inequality follows from (\ref{EE13}) and the last term can be further evaluated by Lemma~\ref {Foster} as
    \begin{align*}
        \frac{n^2\delta^2}{w_{\min}} \Abs{\matw^{1/2} \matb \matz^{\top}}_{F}^{2}
         & =\frac{n^2\delta^{2}}{w_{\min}}
        \sum_{\edge{a}{b} \in E} w_{ab} \Abs{\matz (\vece_{a} - \vece_{b})}^{2}                                                      \\
         & \leq
        \frac{\delta^{2}n^2 (1+\epsilon)}{w_{\min}}
        \sum_{\edge{a}{b} \in E} w_{ab}\|\matw^{\frac{1}{2}} \matb \lap^{\dagger} (\vece_a-\vece_b)\|^2                              \\
         & =
        \frac{\delta^{2}n^2 (1+\epsilon)}{w_{\min}}
        \sum_{\edge{a}{b} \in E} w_{ab}(\vece_a-\vece_b)^\top \lap^{\dagger} \matb^\top \matw \matb \lap^{\dagger} (\vece_a-\vece_b) \\
         & =
        \frac{\delta^{2}n^2 (1+\epsilon)}{w_{\min}}
        \sum_{\edge{a}{b} \in E} w_{ab}(\vece_a-\vece_b)^\top \lap^{\dagger} \lap \lap^{\dagger} (\vece_a-\vece_b)                   \\
         & =
        \frac{\delta^{2}n^2 (1+\epsilon)}{w_{\min}}
        \sum_{\edge{a}{b} \in E} w_{ab}(\vece_a-\vece_b)^\top \lap^{\dagger} (\vece_a-\vece_b)                                       \\
         & =
        \frac{\delta^{2} n^2 (1+\epsilon)}{w_{\min}}
        \sum_{\edge{a}{b} \in E} w_{ab} R_{ab}
        =\frac{\delta^{2} n^2(n-1) (1+\epsilon)}{w_{\min}},
    \end{align*}
    where $w_{ab}$ represents the weight of the edge with endvertices $a$ and $b$.

    In addition, \(\Abs{\matz (\vece_u - \vecpi)}^2\) can also be bounded by
    \begin{equation*}\label{EE17}
        \begin{split}
            \Abs{\matz (\vece_u - \vecpi)}^2
             & \geq(1 - \epsilon) C(u)=(1-\epsilon)(\vece_u-\vecpi)^{\top} \lap^{\dagger} (\vece_u-\vecpi)                       \\
             & \geq (1-\epsilon)\lambda_{n}^{-1} \Abs{\vece_u-\vecpi}^2\geq (1-\epsilon)(n w_{\max})^{-1} \frac{(d-d_u)^2}{d^2}.
        \end{split}
    \end{equation*}
    In the above equation, the first inequation follows due to the following reason.   Note that \(\vece_u-\vecpi\) is orthogonal to  vector \(\vecone\), which is an eigenvector of \(\lap^{\dagger}\) corresponding to the \(0\) eigenvalue. Hence, \((\vece_u-\vecpi)^{\top} \lap^{\dagger} (\vece_u-\vecpi) \ge \lambda_{n}^{-1} \Abs{\vece_u-\vecpi}^2\) always holds true.

    Thus, we have
    \begin{equation*}
        \frac{
            \abs{ \Abs{\matz (\vece_{u} - \vecpi)} -  \Abs{\tilmatz  (\vece_{u} - \vecpi)}}
        }{
            \Abs{\matz (\vece_{u} - \vecpi)}
        }\leq\delta \left(\frac{n^2(n-1) (1+\epsilon)}{w_{\min}}\right)^{1/2}\left(\frac{n
            w_{\max}}{1-\epsilon}\right)^{1/2}\frac{d}{d-d_u}\leq\frac{\epsilon}{3},
    \end{equation*}
    where \(\delta\) is given by (\ref{EE14}).
\end{proof}

Lemma~\ref{lem:error1} leads to the following theorem.
\begin{theorem}
    \label{TheoAlg1}
    There is a \(\Otil(m\log{c}/\epsilon^2)\) time algorithm, which  inputs  a scalar \(0<\epsilon<1\) and a graph \(\gr=(V,E,w)\) where \(c=\frac{w_{\max}}{w_{\min}}\), and returns a \((24\log n/\epsilon^2)\times n\) matrix \(\tilmatz\) such that with probability at least \(1-1/n\),
    \begin{align}
        (1-\epsilon)^2  C(u) \leq \|\tilmatz(\vece_u-\vecpi)\|^2 \leq (1+\epsilon)^2  C(u)\nonumber
    \end{align}
    for any vertex \(u \in V\).
\end{theorem}

Based on Theorem~\ref{TheoAlg1}, we  present an algorithm \textsc{ApproxHK} to approximately compute the walk centrality \(H_u\) for all  vertices  \(u \in V\) and the Kemeny constant \(\kem\), the pseudocode of which is provided in Algorithm~\ref{ALG01}.

\begin{algorithm}
    \caption{\textsc{ApproxHK}\((\gr, \epsilon)\)}
    \label{ALG01}
    \Input{
        \(\gr\): a connected undirected graph. \\
        \(\epsilon\): an approximation parameter 
    }
    \Output{
        \(\tilde{H}=\{u,\tilde{H}_u| u \in V\}\): \(\tilde{H}_u\) is an approximation of the walk centrality \(H_u\) of vertex \(u\);
        \(\approxkem\): approximation of the Kemeny constant \(\kem\)
    }
    \(\lap=\) Laplacian of \(\gr\),\, \(d_u=\) the degree of \(u\) for all \(u\in V\), \(d=\sum_{u \in V} d_u\)\;  Construct a matrix \(\matq_{k \times m}\),  where \(k=\lceil 24\log n/\epsilon^2 \rceil\) and each entry is \(\pm 1/\sqrt{k}\) with identical probability\;
    \For{\(i=1\) to \(k\)}{
        \(\vecq_i^\top \)=the \(i\)-th row of  \(\matq_{k \times m} \matw^{1/2} \matb\)  \\
        \(\tilde{\vecz}_i=\mathtt{Solver}(\lap, \vecq_i, \delta)\) where parameter   \(\delta\) is  given by~(\ref{EE14}) \\
    }
    Calculate the constant vector \({\bf p}=\tilmatz \vecpi\)\;
    \ForEach{\(u\in V\)}{
    \(\tilde{H}_u=d \|\tilmatz_{:,u}-{\bf p}\|^2\)
    }
    \(\approxkem=\sum_{u \in V} \frac{d_u}{d} \tilde{H}_u\)\;
    \Return \(\tilde{H}=\{u,\tilde{H}_u| u \in V\}\) and \(\approxkem\)
\end{algorithm}

\section{Group Walk Centrality}\label{sec:gwc}

In this section, we extend the walk centrality \(H_j\) to the case of a group $S$ of vertices, denoted as \(H(S)\) and called Group Walk Centrality (GWC), which measures the relative importance of the vertex group $S$. We first give the definition of GWC. Then, we provide an examination of GWC, by establishing a connection between GWC and random detour time and the Kemeny constant. Finally, we show that as a set function \(\manc{\cdot}\), GWC is monotonic and supermodular.

\subsection{Definition and Expression}

We utilize absorbing random walks with multiple traps~\cite{ZhYaLi12} to define GWC. For random walks on a connected graph \(\gr=(V,E)\) with vertices in set  \(S\) being traps (absorbing vertices), the hitting time (absorbing time) \(H_{iS}\) from vertex \(i\) to vertex group \(S\) is the expected number of steps for a walker starting from \(i\) to visit any vertex of \(S\) for the first time. If we denote the time steps for a walker starting from \(i\) to first reach an arbitrary vertex in \(S\) as random variable \(T_{iS}\), then we have \(H_{iS}=\mean{T_{iS}}\). Based on the hitting time \(H_{iS}\) to a vertex group $S$, we define the GWC of a set of vertices.

\begin{definition}[Group Walk Centrality, GWC]\label{def:manc}
    For a given vertex group \(S\) in a connected undirected graph \(\gr=(V,E)\), its GWC \(\manc{S}\) is defined as the expected time for a random walker starting from a vertex chosen from \(\gr\) according to the stationary distribution \(\vecpi\), to reach any vertex in \(S\).  Mathematically, \(\manc{S}\) is expressed as follows:
    \begin{equation*}
        \manc{S}=\sum_{i\in V}\pi_i H_{iS}.
    \end{equation*}
\end{definition}
Note that for the case that set \(S\) contains only one vertex \(j\), GWC \(\manc{S}\) reduces to walk centrality \(H_j\). If  \(\manc{S}\) is larger, then the vertices in \(S\) are not easily accessible, which means that the vertices in \(S\) are more peripheral, and are thus less important; if \(\manc{S}\) is smaller, then the vertices in \(S\) are more important.

We now give an expression for \(\manc{S}\). For random walks in graph \(\gr=(V,E)\) with absorbing vertex group \(S\) and transition matrix \(\matp\),  the fundamental matrix \(\matf\) is~\cite{ZhYaLi12}    \begin{equation}\label{eq:funda}
    \matf=\sum_{l=0}^\infty\matp_{-S}^l=(\mati-\matp_{-S})^{-1}=(\mati-\matp)_{-S}^{-1}.
\end{equation}
From~\eqref{eq:funda}, we can see that the entry \(\matf_{[i,j]}\) of matrix \(\matf\) is the expected number of times a random walker starting from vertex \(i\) visits vertex \(j\) before being absorbed by vertices in  \(\manc{S}\). Then, \(\manc{S}\)  can be expressed as
\begin{equation*}
    \manc{S}=\vecpi_{-S}^\top\matf\vecone=\vecpi_{-S}^\top(\mati-\matp)_{-S}^{-1}\vecone.
\end{equation*}
By using~\eqref{eq:trs}, \(\manc{S}\) can be further written as
\begin{equation}\label{eq:GWC}
    \manc{S}  =\vecpi_{-S}^\top\mypar{\mati-\matd^{-1}\mata}_{-S}^{-1}\vecone
    =\vecpi_{-S}^\top\mypar{\mati-\matd^{-1}\mata}_{-S}^{-1}\matd_{-S}^{-1}\matd_{-S}\vecone
    =\vecpi_{-S}^\top\lap_{-S}^{-1}\vecd_{-S}.
\end{equation}
Since $\lap_{-S}$ is an SDDM matrix, according to~\eqref{eq:GWC} and \lemref{lemma:ST}, we can approximately compute GWC of any given vertex group \(S\) by calling \(\mathtt{Solver}\mypar{\lap_{-S},\vecd_{-S},\delta}\). We omit the detail of this approximation algorithm, since  it simply utilizes~\lemref{lemma:ST}.

\subsection{Connection with Other Quantities of Random Walks}

In this subsection, we establish a link between GWC $H(S)$  in graph \(\gr\) and some other quantities related to hitting times for random walks without traps in \(\gr\).

We first introduce a notion of random detour walk $i\rightarrow S \rightarrow j$. It is a random walk starting from a vertex $i$, that must visit at least a vertex in set $S$ before it arrives at the target vertex $j$ and stops. The group random detour time  \(D_{ij}(S)\) for the random detour walk $i\rightarrow S \rightarrow j$ is defined as the expected time for the walker starting from vertex \(i\) to visit at least a vertex in group \(S\), and then reaches vertex \(j\) for the first time. Note that when set $S$ includes only one vertex $k$, $i\rightarrow S \rightarrow j$ and \(D_{ij}(S)\) are reduced to the corresponding notion and quantity previously proposed in~\cite{RaZh13,BoRaZh11} for a single vertex. By definition, for a single vertex $k$, the random detour walk for $i\rightarrow k \rightarrow j$ is equal to $H_{ik}+ H_{kj} - H_{ij} $, which is defined based on the following triangular inequality  for hitting time $ H_{ij} \leq H_{ik}+ H_{kj}$ of all $1 \leq i,j,k \leq n$~\cite{Hu05}.  An intuitive interpretation of the random detour time $H_{ik}+ H_{kj} - H_{ij}$ is that the number of transition steps $H_{ij}$
from one vertex $i$ to another vertex $j$ can be expected to increase when a forced passage through a third vertex $k$ is prescribed.

For an absorbing random walk in graph \(\gr=(V,E)\)  with vertices in set \(S\)  being absorbing ones,  we denote the probability that a walker starting from vertex \(i\) first reaches vertex \(u\) in absorbing group \(S\) as the \(\myord{(i,u)}\) entry of matrix \(\matpdistr\in\rea^{n\times\abs{S}}\). Then, \(D_{ij}(S)\) can be expressed as
\begin{equation}
    D_{ij}(S)= H_{iS}+\sum_{k=1}^{\abs{S}}\matpdistr_{[i,k]}H_{kj}.
\end{equation}
This expression shows that  if \(D_{ij}(S)\) is larger, then it is difficult for a walker to reach vertices in \(S\), which indicates that vertices in \(S\) are relatively less significant in the overall network. Thus, \(D_{ij}(S)\) is consistent with the quantity $H_{iS}$. This enlightens us to explore the relation between \(\manc{S}\) and related quantities about hitting times, such as  \(D_{ij}(S)\). Theorem~\ref{thm:connection-multiple} establishes a link governing GWC \(\manc{S}\), group random detour time \(D_{ij}(S)\), and the Kemeny constant \(\kem\).

\begin{theorem}\label{thm:connection-multiple}
    For an arbitary vertex group \(S\subseteq V\) in a connected graph \(\gr=(V,E)\) with \(n\) vertices,
    \begin{equation}\label{eq:connection-multiple}
        \manc{S}+\kem=\sum_{i=1}^n\sum_{j=1}^n\pi_i\pi_jD_{ij}(S).
    \end{equation}
\end{theorem}
\begin{proof}
    According to \lemref{lem:non-absorb-fund}, the hitting time can be represented in terms of \(\matfstar\). Thus, the right-hand side of \eqref{eq:connection-multiple} can be rewritten as
    \begin{equation}\label{eq:detour}
        \begin{split}
            \sum_{i=1}^n\sum_{j=1}^n\pi_i\pi_jD_{ij}(S)
             & =\sum_{i=1}^n\pi_i H_{iS}+\sum_{i=1}^n\sum_{j=1}^n\sum_{k=1}^{\abs{S}}\pi_i\pi_j\matpdistr_{[i,k]}H_{kj}               \\
             & = \manc{S}+\sum_{j=1}^n\pi_j\vecpi^\top\matpdistr\mypar{\matfstar_{[j,j]}\vecone-\matfstar_{[S,j]}}                    \\
             & = \manc{S}+\sum_{j=1}^n\pi_j\matfstar_{[j,j]}\vecpi^\top\matpdistr\vecone-\vecpi^\top\matpdistr\matfstar_{[S,:]}\vecpi \\
             & = \manc{S}+\sum_{j=1}^n\pi_j\matfstar_{[j,j]}-1,
        \end{split}
    \end{equation}
    where the final equality is due to the fact that \(\matpdistr\vecone=\vecone\) and \(\matfstar\vecpi=\vecone\)~\cite{BoRaZh11}.

    On the other hand, we can rewrite the Kemeny constant \(\kem\) as
    \begin{equation}\label{eq:kemeny}
        \kem  =\sum_{i=1}^n\sum_{j=1}^n\pi_i\pi_jH_{ij}=\sum_{i=1}^n\sum_{j=1}^n\pi_i\pi_j\mypar{\matfstar_{[j,j]}-\matfstar_{[i,j]}}=\sum_{j=1}^n\pi_j\matfstar_{[j,j]}-1.
    \end{equation}
    Combining~\eqref{eq:detour} and~\eqref{eq:kemeny} completes the proof.
\end{proof}

Since the Kemeny constant \(\kem\) is only related to the structure of graph \(\gr\),~\eqref{eq:connection-multiple} indicates that there exists only a constant difference between GWC \(\manc{S}\) and  a weighted sum of the group random detour time \(D_{ij}(S)\) over all pairs of vertices. Thus, from the perspective of decimating power of vertex group,
GWC \(\manc{S}\) is equivalent to weighted sum of the group random detour times. In this sense, Theorem~\ref{thm:connection-multiple} deepens the understanding of the importance of vertex group \(S\) from a different perspective.

\subsection{Monotonic and Supermodular Properties}

In this subsection, we show that the set function \(\manc{\cdot}\) has some desirable properties. For this purpose, we first introduce the definitions of monotone and supermodular set functions. For simplicity, we use \(S+u\) to represent \(S\cup\setof{u}\).

\begin{definition}[Monotonicity]
    A set function \(f:2^{V}\to\rea^+\) is monotonic if and only if the inequality \(f(X)\ge f(Y)\) holds for two arbitrary  nonempty sets \(X\) and \(Y\)  satisfying \(X\subseteq Y\subseteq V\).
\end{definition}

\begin{definition}[Supermodularity]
    A set function \(f:2^{V}\to\rea^+\) is  supermodular if and only if the inequality \(f(X)-f(X+u)\ge f(Y)-f(Y+u)\) holds for two arbitrary nonempty sets \(X\) and \(Y\) satisfying \(\forall X\subseteq Y\subseteq V, u\in V\setminus Y\).
\end{definition}

Next we demonstrate that the set function \(\manc{\cdot}\) is  monotonically decreasing and supermodular.

\begin{theorem}\label{thm:mono}
    Let \(S\) be a nonempty set of vertices of a connected weighted undirected graph \(\gr=(V,E,w)\). Then, for any vertex \(u\in V\setminus S\),
    \begin{equation*}
        \manc{S}\ge \manc{S+u}.
    \end{equation*}
\end{theorem}
\begin{proof}
    By \defref{def:manc}, when a random walker arrives at a vertex in \(S\), it must have arrived at vertices in \(S+u\). Therefore, \(H_{vS}\ge H_{v\mypar{S+u}}\) holds for any vertex \(v\in V\), which concludes our proof.
\end{proof}

\begin{theorem}\label{thm:supermod}
    Let \(S\) and \(T\) be two  nonempty sets of vertices of a connected weighted undirected graph \(\gr=(V,E,w)\) obeying  \(S\subseteq T\subsetneq V\). Then, for any vertex \(u\in V\setminus T\),
    \begin{equation*}
        \manc{S}-\manc{S+u}\ge \manc{T}-\manc{T+u}.
    \end{equation*}
\end{theorem}
\begin{proof}
    By definition of absorbing time for random walks with traps, \(\manc{S}-\manc{S+u}\) can be represented as
    \begin{equation}\label{eq:supermod-diff}
        \begin{split}
            \manc{S}-\manc{S+u}
             & = \sum_{v\in V}\pi_v\mypar{\mean{T_{vS}}-\mean{T_{v\mypar{S+u}}}}           \\
             & = \sum_{v\in V}\pi_v\mypar{\mean{T_{vS}}-\mean{\min\setof{T_{vS},T_{vu}}}}.
        \end{split}
    \end{equation}
    We next represent \(\min\setof{T_{vS},T_{vu}}\) in terms of the indicator function. For a subset \(A\) of the probability space, we denote \(\indfunc{A}\) as the indicator function of the event \(A\). Concretely, \(\indfunc{A}\) is defined as
    \begin{equation*}
        \indfunc{A}=
        \begin{cases}
            1, & \text{event }A\text{ holds true;} \\
            0, & \text{otherwise.}
        \end{cases}.
    \end{equation*}
    Then, \(\min\setof{T_{vS},T_{vu}}\) can be recast as
    \begin{equation*}
        \min\setof{T_{vS},T_{vu}}=T_{vS}\indfunc{\setof{T_{vS}\le T_{vu}}}+T_{vu}\indfunc{\setof{T_{vu}<T_{vS}}},
    \end{equation*}
    inserting which into~\eqref{eq:supermod-diff}  and considering  the linearity of the mean, one obtains
    \begin{equation*}
        \manc{S}-\manc{S+u}=\sum_{v\in V}\pi_v\mean{(T_{vS}-T_{vu})\indfunc{\setof{T_{vu}<T_{vS}}}}.
    \end{equation*}
    Similarly, we have
    \[\manc{T}-\manc{T+u}=\sum_{v\in V}\vecpi_v\mean{(T_{vT}-T_{vu})\indfunc{\setof{T_{vu}<T_{vT}}}}.\]
    As mentioned in proof of \thmref{thm:mono}, \(T_{vS}\ge T_{vT}\) holds for any vertex \(v\in V\).
    Thus, when \(T_{vu}<T_{vT}\) holds true, the relation \(T_{vu}<T_{vS}\) also holds.  Then, \(\indfunc{\setof{T_{vu}<T_{vS}}}\ge\indfunc{\setof{T_{vu}<T_{vT}}}\) holds for any vertex \(v\in V\).
    Combining these two inequalities, one obtains
    \[(T_{vS}-T_{vu})\indfunc{\setof{T_{vu}<T_{vS}}}\ge(T_{vT}-T_{vu})\indfunc{\setof{T_{vu}<T_{vT}}}\]
    for any vertex \(v\in V\). This completes the proof.
\end{proof}

\section{Minimizing Group Walk Centrality Subject to Cardinality Constraint}\label{sec:mingwc}

In this section, we define and study the problem of minimizing the GWC \(\manc{S}\) by optimally selecting a fixed number $k=\abs{S}$ vertices constituting set.

\subsection{Problem Formulation}

As shown above, \(\manc{S}\) is a monotonically decreasing set function, since the addition of any vertex to set $S$ will lead to a decrease of the quantity \(\manc{S}\). Then, we naturally raise the following problem called MinGWC: How to optimally choose a set \(S\) of \(k=\abs{S}\) vertices in graph \(\gr=(V,E)\), so that the quantity \(\manc{S}\) is minimized. Mathematically, the optimization problem subject to a cardinality constraint can be formally stated as follows.

\begin{problem}[\underline{Min}imization of \underline{G}roup \underline{W}alk \underline{C}entrality, MinGWC]
Given a connected weighted undirected graph \(\gr=(V,E,w)\) with \(n\) vertices, \(m\) edges and an integer \(k\ll n\), the goal is to find a vertex group \(S^*\subseteq V\) such that the GWC \(\manc{S^*}\) is minimized. In other words,
\begin{equation*}
    S^*=\argmin_{S\subseteq V,\abs{S}=k}\manc{S}.
\end{equation*}
\end{problem}

Since the problem MinGWC is inherently combinatorial, we can naturally resort to a na\"{\i}ve solution by exhausting all $\tbinom{n}{k}$ cases of vertex selection. For each of the possible cases that $k$ vertices are selected to form set $S$, we calculate the quantity \(\manc{S}\) corresponding to the $k$ vertices in set $S$ and output the optimal solution $S^*$ with the minimum group walk centrality. Evidently, this method fails when $n$ or $k$ is slightly large since it has an exponential time complexity $O(\tbinom{n}{k}n^3)$, assuming that inverting an $n \times n$ matrix needs $O(n^3)$  time. Below we will show that the  MinGWC problem is in fact NP-hard.

\subsection{Hardness of the MinGWC Problem}

In this section, we prove that  the MinGWC problem is NP-hard, by giving a reduction from the vertex cover problem for 3-regular graphs, the vertices of which have identical degree 3. For a 3-regular graph, the vertex cover problem is NP-complete~\cite{FrHeJa98}, the decision version of which is stated below.
\begin{problem}\label{Cover3}
Given a connected 3-regular graph \(G=(V,E)\) and an integer \(k\),  the vertex cover problem is to determine whether exists a subset \(S\subseteq V\) of \(\abs{S}\le k\) vertices, such that every edge in \(E\) is incident with at least one vertex in \(S\), i.e., \(S\) is a vertex cover of \(G=(V,E)\).
\end{problem}

We continue to give a decision version of the MinGWC problem for a 3-regular graph.
\begin{problem}\label{MinGWC3}
Given a connected 3-regular graph \(G=(V,E)\) with every edge having a unit weight,  an integer \(k\), and a threshold $(n-k)/n$, the MinGWC problem  asks whether exists a subset \(S\subseteq V\) of \(\abs{S}\le k\) vertices, such that \(\manc{S}\ge (n-k)/n\).
\end{problem}

\begin{theorem}\label{thm:np-hard}
    The MinGWC problem is NP-hard.
\end{theorem}
\begin{proof}
    We give a reduction from the above vertex cover problem in a 3-regular graph  \(G=(V,E)\) to the decision version of the problem of group walk centrality. Specifically, we will prove that for any nonempty set \(S\subseteq V\) of  \(\abs{S}=k\)  vertices, \(\manc{S}\ge (n-k)/n\), where the equality holds if and only if \(S\) is a vertex cover of \(\gr\).

    We first prove that if \(S\) is a vertex cover of \(\gr\), then \(\manc{S}=(n-k)/n\). In the case that \(S\) is a vertex cover of \(\gr\), there is no edge between vertices in \(V\setminus S\). Then~\eqref{eq:GWC} is simplified as
    \begin{equation*}
        \manc{S}=\vecpi_{-S}^\top\lap_{-S}^{-1}\vecd_{-S}=\vecpi_{-S}^\top\matd_{-S}^{-1}\vecd_{-S}=\vecpi_{-S}^\top\vecone.
    \end{equation*}
    Since \(\gr\) is a 3-regular graph, \(\vecpi=\begin{pmatrix}1/n&\cdots&1/n\end{pmatrix}^\top\).
    Thus, \(\manc{S}=\vecpi_{-S}^\top\vecone=(n-k)/n\).

    We continue to demonstrate that if \(S\) is not a vertex cover of \(\gr\), then \(\manc{S}>(n-k)/n\).  Note that for the case that \(S\) is not a vertex cover,  for a vertex \(u\in S\) and any vertex \(v\in V\setminus S\), \(H_{uS}=0\) and \(H_{vS}\ge1\) always hold. Moreover, there exists at least one edge \((u',v')\) in edge set \(E\) that is not covered by \(S\). Thus, for a walker starting from vertex \(u'\), it moves to vertex \(v'\) with a probability of at least \(\frac{1}{d_{\max}}>\frac{1}{n}\), then it continue walking until being absorbed by vertices in  \(S\).  For this case the absorbing time $H_{u'S}$ is bounded as
    \begin{equation*}
        H_{u'S}>\bigpar{1-\frac{1}{d_{\max}}}\times1+\frac{1}{d_{\max}}\times2
        =\bigpar{1+\frac{1}{d_{\max}}}>1+\frac{1}{n}.
    \end{equation*}
    Similarly, we obtain \(H_{v'S}>1+\frac{1}{n}\).

    Let \(S'\) denote \(S\cup\setof{u',v'}\). Then we are able to get a lower bound of GWC $\manc{S}$ as
    \begin{align*}
        \manc{S} & =\sum_{v\in V\setminus S'}\pi_vH_{vS}+\pi_{u'}H_{u'S}+\pi_{v'}H_{v'S}
        >\vecpi_{-S'}\vecone+\pi_{u'}\bigpar{1+\frac{1}{n}}+\pi_{v'}\bigpar{1+\frac{1}{n}} \\
                 & =\vecpi_{-S}^\top\vecone+\frac{\pi_{u'}+\pi_{v'}}n
        >\vecpi_{-S}^\top\vecone=(n-k)/n.
    \end{align*}
    This finishes the proof.
\end{proof}

\subsection{A Simple Greedy Algorithm}

Since the MinGWC problem is NP-hard, it is nearly impossible to design a polynomial time algorithm to solve it. One often resorts to greedy algorithm. Due to the proved fact that the objective function \(\manc{\cdot}\) of  the MinGWC problem is monotonic and supermodular, a na\"{\i}ve greedy strategy for vertex selection yields an approximated   solution  with a $(1-\frac{k}{k-1}\frac{1}{e})$ approximation factor~\cite{NeWoFi78}. Initially, we set $S$ to be empty, then $k$ vertices from set $V \setminus S$ are added to $S$ iteratively. During the first iteration, the vertex $u$ with the minimum walk centrality $H_u$ is chosen to be added to $S$. Then, in each iteration step $i>1$, the vertex $u$ in the set $V \setminus S$ of candidate vertices is selected, which has the maximum marginal gain $\Delta(u,S) \defeq \manc{S}-\manc{S+u}$.  The algorithm terminates when $k$ vertices are chosen to be added to set $S$.

At each iteration of the na\"{\i}ve greedy algorithm, one needs computing the maximum marginal gain $\Delta(u,S)$ for every candidate vertex in set $V \setminus S$. As shown in~\eqref{eq:GWC}, a straightforward calculation of  $\Delta(u,S)$ for every vertex \(u\) requires inverting a matrix, which needs $O((n-\abs{S})^3)$ time.  Thus, for the $(O(n-\abs{S}))$ candidate vertices, the computation time at each iteration is $O((n-\abs{S})^4)$. Hence, the total computation complexity of the na\"{\i}ve greedy algorithm is \(O(kn^4)\) since $\abs{S} \ll n$.

In order to further decrease the time complexity of the na\"{\i}ve greedy algorithm, we provide an expression for \(\Delta(u,S)\).
\begin{lemma}\label{lem:delta}
    Given a connected weighted undirected graph \(\gr=(V,E,w)\), the nonempty vertex group \(S\), and the Laplacian matrix \(\lap\), for any vertex \(u\in V\setminus S\),  \(\Delta(u,S)\) can be alternatively expressed as
    \begin{equation}\label{eq:delta}
        \Delta(u,S)=   \manc{S}-\manc{S+u}=\frac{\mypar{\vece_u^\top\lap_{-S}^{-1}\vecd_{-S}}^2}{d\mypar{\vece_u^\top\lap_{-S}^{-1}\vece_u}}.
    \end{equation}
\end{lemma}
\begin{proof}
    For vertex \(u\in V\setminus S\), we denote the \(\myord{u}\) column of submatrix \(\lap_{-S}\) as \(\begin{pmatrix}d_u\\-\veca\end{pmatrix}\). After properly adjusting the vertex labels, matrix \(\lap_{-S}\) can be rewritten  in block form as
    \begin{equation*}
        \lap_{-S}=\begin{pmatrix}d_u & -\veca^\top \\-\veca&\mata\end{pmatrix},
    \end{equation*}
    where \(\mata\) denotes \(\lap_{-(S+u)}\).
    According to the properties of block matrices, we  get
    \begin{equation*}
        \lap_{-S}^{-1}=
        \begin{pmatrix}
            t & t\veca^\top\mata^{-1} \\t\mata^{-1}\veca&\mata^{-1}+t\mata^{-1}\veca\veca^\top\mata^{-1}
        \end{pmatrix},
    \end{equation*}
    where \(t=\pinv{d_u-\veca^\top\inv{\mata}}\).
    Therefore, when \(S\neq\varnothing\), \(\Delta(u,S)\) can be rewritten as
    \begin{equation*}
        \begin{split}
            \Delta(u,S)
             & =\manc{S}-\manc{S+u}
            = \vecpi_{-S}^\top\lap_{-S}^{-1}\vecd_{-S}-\vecpi_{-(S+u)}^\top\lap_{-(S+u)}^{-1}\vecd_{-(S+u)}                                                                                                                                            \\
             & = \frac1t\mypar{t\vecpi_u+t\vecpi_{-(S+u)}^\top\mata^{-1}\veca}\mypar{td_u+t\veca^\top\mata^{-1}\vecd_{-(S+u)}}                                                                                                                         \\
             & =\frac{\mypar{\vecpi_{-S}^\top\lap_{-S}^{-1}\vece_u}\mypar{\vece_u^\top\lap_{-S}^{-1}\vecd_{-S}}}{\vece_u^\top\lap_{-S}^{-1}\vece_u}=\frac{\mypar{\vece_u^\top\lap_{-S}^{-1}\vecd_{-S}}^2}{d\mypar{\vece_u^\top\lap_{-S}^{-1}\vece_u}},
        \end{split}
    \end{equation*}
    finishing the proof.
\end{proof}

\lemref{lem:delta} demonstrate that for any empty set $S$ if \(\lap_{-S}\)  is already computed, by using~\eqref{eq:delta}  we can calculate \(\Delta(u,S)\) for any vertex \(u\in V\setminus S\) in $O((n-\abs{S})^2)$ time.
This avoids directly inverting the matrix \(\lap_{-(S+u)}\), which needs time $O((n-\abs{S}+1)^3)$. When \(S=\varnothing\),~\eqref{eq:delta} becomes invalid.  For this case, we use the result \(H_u=d\mypar{\vece_u-\vecpi}^\top\lap^{\dagger}\mypar{\vece_u-\vecpi}\) in~\lemref{HjK} to compute the walk centrality $H_u$ for every vertex $u \in V$, and add the vertex $u$ with the smallest $H_u$ to $S$.  Based on~\lemref{HjK} and~\lemref{lem:delta}, we propose a simple greedy algorithm \textsc{DeterMinGWC}$(\gr,k)$ to solve the MinGWC problem, which is  outlined in~\algoref{algo:exact-gwcm}. The algorithm first computes $H_u$ for every vertex $u \in V$ in time $O(n^3)$. Then, it works in $k-1$ rounds. In each round, it evaluates  \(\Delta(u,S)\) for all the $(n-\abs{S})$ vertices in set $V \setminus S$ (Line 6)  in $O((n-\abs{S})^3)$ time. Therefore, the whole running time of~\algoref{algo:exact-gwcm} is \(O(kn^3)\),  much lower than the time complexity   \(O(kn^4)\) of the  na\"{\i}ve greedy algorithm.

\begin{algorithm}
    \caption{\textsc{DeterMinGWC}\((\gr,k)\)}
    \label{algo:exact-gwcm}
    \Input{
        A connected weighted undirected graph \(\gr=(V,E,w)\);
        an integer \(k \ll \abs{V}\)
    }
    \Output{\(S_k\): A subset of \(V\) with \(\abs{S_k}=k\)}
    Compute \(\lap\) and \(\vecd\)\;
    \(d\gets\vecd\vecone\),\(\vecpi\gets \inv{d}\vecd\)\;
    \(S_1\gets\setof{\argmin_{u\in V}\setof{d(\vece_u-\vecpi)^\top\lap^\dagger(\vece_u-\vecpi)}}\)\;
    \For{\(i=2,3,\dots,k\)}{
        \ForEach{\(u\in V\setminus S\)}{
            \(\Delta(u,S)\gets\frac{(\vece_u^\top\lap_{-S}^{-1}\vecd_{-S})^2}{d(\vece_u^\top\lap_{-S}^{-1}\vece_u)}\)\;
        }
        \(u^*\gets\argmax_{u\in V\setminus S}\setof{\Delta(u,S)}\)\;
        \(S_i\gets S_{i-1}\cup\setof{u^*}\)
    }
    \Return{\(S_k\)}
\end{algorithm}

In addition to good efficiency, \algoref{algo:exact-gwcm} is very effective since it has a $(1-\frac{k}{k-1}\frac{1}{e})$ approximation factor, as shown in the following theorem.
\begin{theorem}
    The algorithm $S_k$= \textsc {DeterMinGWC}$ (\gr,k)$ takes a connected weighted undirected graph \(\gr=(V,E,w)\) and a positive integer \(1 \geq k \ll n\), and returns a  group \(S_k\) of \(k\) vertices.
    The return vertex group \(S\) satisfies
    \[\manc{\setof{u^*}}-\manc{S_k}\ge\bigpar{1-\frac{k}{k-1}\frac{1}{e}}\mypar{\manc{\setof{u^*}}-\manc{S^*}},\]
    where
    \[u^*\defeq\argmin_{u\in V}\manc{\setof{u}}\] and \[S^*\defeq\argmin_{\abs{S}=k}\manc{S}.\]
\end{theorem}
\begin{proof}
    According to the supermodularity of \(\manc{\cdot}\),
    \[\manc{S_i}-\manc{S_{i+1}}\ge\frac{1}{k}\mypar{\manc{S_i}-\manc{S^*}}\]
    holds for any positive integer \(i\), which indicates
    \[\manc{S_{i+1}}-\manc{S^*}\le\bigpar{1-\frac{1}{k}}\mypar{\manc{S_i}-\manc{S^*}}.\]
    Then, we  further obtain
    \begin{equation*}
        \manc{S_k}-\manc{S^*}\le \bigpar{1-\frac{1}{k}}^{k-1}\mypar{\manc{S_1}-\manc{S^*}}
        \le                      \frac{k}{k-1}\frac{1}{e}\mypar{\manc{S_1}-\manc{S^*}},
    \end{equation*}
    which, together with the fact that \(S_1=\setof{u^*}\), completes the proof.
\end{proof}

\section{Fast Greedy Algorithm for the MinGWC Problem}\label{sec:approx-algo2}

Although \algoref{algo:exact-gwcm} is much faster than the  na\"{\i}ve greedy  algorithm, it  still cannot handle  large networks since its time complexity is  \(O(kn^3)\). As shown above, the key steps for \algoref{algo:exact-gwcm} are  to compute the walk centrality $H_u$ for every vertex $u \in V$ and the marginal gain  \(\Delta(u,S)\) for every vertex $u \in (V \setminus S)$. In this section, we provide efficient approximations for  $H_u$ and \(\Delta(u,S)\). For $H_u$, it can be efficiently evaluated by using \algoref{ALG01}. For \(\Delta(u,S)\), we can alternatively approximate the two quantities  $\vece_u^\top \lap_{-S}^{-1} \vecd_{-S}$ and $\vece_u^\top \lap_{-S}^{-1}\vece_u$ in the numerator and denominator of~\eqref{eq:delta}, respectively.

In the rest of this section, we first approximate $\vece_u^\top \lap_{-S}^{-1} \vecd_{-S}$ and $\vece_u^\top \lap_{-S}^{-1}\vece_u$. Then, we present an  error-guaranteed approximation to  \(\Delta(u,S)\), based on which we further propose  a fast approximation algorithm to the MinGWC problem  with $(1-\frac{k}{k-1} \frac{1}{e}-\epsilon)$ approximation ratio and time complexity $\tilde{O}(km\epsilon^{-2})$ for any error parameter $0< \epsilon < 1$.

\subsection{Approximation of the Quantity $\mathbf{e}_i^\top \mathbf{L}_{-S}^{-1}\mathbf{d}_{-S}$}

We first  approximate the quantity \(\vece_i^\top\lap_{-S}^{-1}\vecd_{-S}\) for any node $i\in V\setminus S$. To this end, we introduce the following two lemmas.
\begin{lemma}\label{lem:norm-ineq}
    Let \(\gr=(V,E,w)\) be a connected weighted undirected graph with Laplacian matrix  \(\lap\). Then, for an arbitrary subset \(S\subseteq V\) of vertices and an arbitrary vector \(\vecv\in\rea^{n-\abs{S}}\), the relation \(v_i^2\le nw_{\min}^{-1}\Abs{\vecv}^2_{\lap_{-S}}\) holds.
\end{lemma}
\begin{proof}
    As shown in~\secref{sub:lap}, Laplacian matrix $\lap$ can be written as \(\lap=\matb^\top\matw\matb\).   Similarly, as an SDDM matrix, \(\lap_{-S}\) can be decomposed as \(\lap_{-S}=\lap'+\matz=\matb'^\top\matw'\matb'+\matz\), where \(\lap'\) denotes the Laplacian matrix of a  graph, and \(\matz\) is a diagonal matrix with every diagonal entry  larger than 0. In order to prove the lemma, we distinguish two cases:  \(\matz_{[i,i]}\neq0\) and  \(\matz_{[i,i]}=0\).

    For the first case \(\matz_{[i,i]}\neq0\),  we have \(\matz_{[i,i]}\ge w_{\min}\). Then, it is apparent that \(v_i^2\le nw_{\min}^{-1}\Abs{\vecv}^2_{\lap_{-S}}\). For the second case  \(\matz_{[i,i]}=0\), there must exist a vertex \(j\)  satisfying \(\matz_{[j,j]}\ge w_{\min}\) in a component of the resultant graph containing vertex \(i\), which is obtained from \(\gr=(V,E,w)\) after removing vertices in \(S\).    Let \(\pscr_{ij}\) denotes a simple path connecting vertex \(i\) and vertex \(j\). Then, we have
    \begin{align*}
        \Abs{\vecv}^2_{\lap_{-S}}
         & \ge w_{\min}\sum_{\edge{a}{b}\in\pscr_{ij}}(v_a-v_b)^2+w_{\min}v_j^2                             \\
         & \ge w_{\min}n^{-1}\mypar{\sum_{\edge{a}{b}\in\pscr_{ij}}(v_a-v_b)+v_j}^2\ge w_{\min}n^{-1}v_i^2,
    \end{align*}
    which completes our proof.
\end{proof}

\begin{lemma}\label{lem:trace-lap}
    Let \(\gr=(V,E,w)\) be a connected weighted undirected graph with \(n\) vertices and the minimum edge weight $ w_{\min}$, and let \(\lap\) be the Laplacian matrix of \(\gr\). Then, for any nonempty set \(S\subseteq V\),
    \begin{equation*}
        \trace{\lap_{-S}^{-1}}\le n^2w_{\min}^{-1}.
    \end{equation*}
\end{lemma}
\begin{proof}
    According to the result in~\cite{ClPo11}, the diagonal entry \(\mypar{\lap_{-S}}_{[u,u]}\) of \(\lap_{-S}\) is exactly the resistance distance \(R_{uS}\) between vertex \(u\) and vertex group \(S\).  Since \(\gr\) is connected, the length of an arbitrary simple path connecting vertex \(u\) any vertex in \(S\) is no more than \(n\).
    Making use of~\eqref{eq:phy-resist}, we have
    \begin{equation*}
        \trace{\lap_{-S}^{-1}} =\sum_{i=1}^n\vece_i^\top\lap_{-S}^{-1}\vece_i=\sum_{i=1}^nR_{iS}\le\sum_{i=1}^n\sum_{\edge{a}{b}\in\pscr_{iS}}w_{ab}^{-1}\le n^2w_{\min}^{-1},
    \end{equation*}
    which finishes the proof.
\end{proof}

With the above two lemmas, we can approximate \(\vece_u^\top\lap_{-S}^{-1}\vecd_{-S}\). For simplicity of the following proofs, define $\vecx_i $ as \(\vecx_i=\vece_i^\top\lap_{-S}^{-1}\vecd_{-S}\) for any vertex \(i\in V\setminus S\). By definition, $\vecx_i $ represents the trapping time of a random walker starting at vertex \(i\) and being absorbed by  vertices in $S$. Then, it follows that \(\vecx_i\ge1\). Since \(\lap_{-S}\) is an SDDM matrix, we can estimate $\vecx_i $ by applying~\lemref{lemma:ST}, which avoids computing the inverse of matrix \(\lap_{-S}^{-1}\).

\begin{lemma}\label{lem:approx-numer}
    Let \(\gr=(V,E,w)\) be a connected weighted undirected graph with the Laplacian matrix \(\lap\). Let  \(S\) be a nonempty subset of vertex set $V$, let $\epsilon$ an error parameter in  $(0,1)$, and let \(\vecx'=\mathtt{Solver}(\lap_{-S},\vecd_{-S},\delta_1)\), where \(\delta_1=\frac{w_{\min}\epsilon}{7n^3w_{\max}}\). Then, for any vertex \(i\in V\setminus S\),
    \begin{equation}\label{eq:approx-numer}
        \vecx_i=\vece_i^\top\lap_{-S}^{-1}\vecd_{-S}\approx_{\epsilon/7}\vecx'_i.
    \end{equation}
\end{lemma}
\begin{proof}
    Combining \lemref{lemma:ST} and \lemref{lem:norm-ineq}, \((\vecx'_i-\vecx_i)^2\) is bounded  as
    \begin{align*}
        (\vecx'_i-\vecx_i)^2
         & \le nw_{\min}^{-1}\Abs{\vecx'-\vecx}^2_{\lap_{-S}}
        \le nw_{\min}^{-1}\delta_1^2\Abs{\vecx}^2_{\lap_{-S}}             \\
         & \le n^4w_{\min}^{-1}w_{\max}^2\delta_1^2\trace{\lap_{-S}^{-1}}
        \le n^6w_{\min}^{-2}w_{\max}^2\delta_1^2,
    \end{align*}
    where the last inequality is due to \lemref{lem:trace-lap}.
    Considering \(\vecx_i\ge1\), we obtain
    \begin{equation*}
        \frac{\abs{\vecx'_i-\vecx_i}}{\vecx_i}\le n^3w_{\max}w_{\min}^{-1}\delta_1=\epsilon/7,
    \end{equation*}
    which leads to the desirable  result.
\end{proof}


\subsection{Approximation of the Quantity $\mathbf{e}_i^\top \mathbf{L}_{-S}^{-1}\mathbf{e}_i$}

We next estimate the quantity $\mathbf{e}_i^\top \mathbf{L}_{-S}^{-1}\mathbf{e}_i$ in the denominator  of~\eqref{eq:delta}. For this purpose,
we represent $\mathbf{e}_i^\top \mathbf{L}_{-S}^{-1}\mathbf{e}_i$ in terms of an Euclidean norm as
\begin{equation}\label{eq:denom-solver}
    \begin{split}
        \vece_i^\top\lap_{-S}^{-1}\vece_i
         & =\vece_i^\top\lap_{-S}^{-1}(\matb'^\top\matw'\matb'+\matz)\lap_{-S}^{-1}\vece_i
        = \vece_i^\top\lap_{-S}^{-1}\matb'^\top\matw'\matb'\lap_{-S}^{-1}\vece_i+\vece_i^\top\lap_{-S}^{-1}\matz\lap_{-S}^{-1}\vece_i \\
         & = \Abs{\matw'^{1/2}\matb'\lap_{-S}^{-1}\vece_i}^2+\Abs{\matz^{1/2}\lap_{-S}^{-1}\vece_i}^2.
    \end{split}
\end{equation}
Thus, the evaluation of $\mathbf{e}_i^\top \mathbf{L}_{-S}^{-1}\mathbf{e}_i$ is reduced to computing  the $\ell_2$-norm of two vectors in $\rea^m$ and $\rea^n$, respectively. However, exactly calculating $\mathbf{e}_i^\top \mathbf{L}_{-S}^{-1}\mathbf{e}_i$ by evaluating $\ell_2$-norms still has a high computational cost, since we still need to invert the matrix  $\mathbf{L}_{-S}$. Although by utilizing \lemref{lemma:ST}, one can approximate $\ell_2$-norms. Unfortunately,  since the dimensions of matrices \(\matw'\) and \(\matz\) are, respectively,  $O(m)$ and  $O(n)$, resulting in an unacceptable number of calls to the \(\mathtt{Solver}\) function. In the following text, by using~\lemref{lemma:JL} we will propose an approximation algorithm to estimate the two $\ell_2$-norms, which significantly lessens the computational cost.

In order to exploit \lemref{lemma:JL}, we introduce two projection matrices \(\matq\in\rea^{q\times m}\) and \(\matr\in\rea^{r\times n}\), where \(q=r=\ceil{24\epsilon^{-2}\log n}\). Each of the entry for \(\matq\in\rea^{q\times m}\) and \(\matr\in\rea^{r\times n}\) is  $1/\sqrt{q}$ or $1/\sqrt{q}$ with identical probability $1/2$. Then, we have
\begin{equation}\label{eq:denom-jl}
    \mathbf{e}_i^\top \mathbf{L}_{-S}^{-1}\mathbf{e}_i =\Abs{\matw'^{1/2}\matb'\lap_{-S}^{-1}\vece_i}^2+\Abs{\matz^{1/2}\lap_{-S}^{-1}\vece_i}^2\approx_\epsilon\Abs{\matq\matw'^{1/2}\matb'\lap_{-S}^{-1}\vece_i}^2+\Abs{\matr\matz^{1/2}\lap_{-S}^{-1}\vece_i}^2.
\end{equation}
For the sake of simplicity, we denote \(\matw'^{1/2}\matb'\lap_{-S}^{-1}\) and \(\matq\matw'^{1/2}\matb'\lap_{-S}^{-1}\) as \(\matx\) and \(\tilde{\matx}\), respectively. Moreover, we  denote \(\matz^{1/2}\lap_{-S}^{-1}\) and \(\matr\matz^{1/2}\lap_{-S}^{-1}\) as \(\maty\) and \(\tilde{\maty}\), respectively.
Then~\eqref{eq:denom-jl} can be rephrased as
\begin{equation}\label{eq:denom-solver-simp}
    \vece_i^\top\lap_{-S}^{-1}\vece_i =\Abs{\matx\vece_i}^2+\Abs{\maty\vece_i}^2\approx_\epsilon\Abs{\tilde{\matx}\vece_i}^2+\Abs{\tilde{\maty}\vece_i}^2.
\end{equation}
Making this expression,  we can efficiently approximate the quantity $\vece_i^\top\lap_{-S}^{-1}\vece_i$.

\begin{lemma}\label{lem:approx-denom}
    Let \(\gr=(V,E,w)\) be a connected weighted undirected graph with Laplacian matrix \(\lap\). Let  \(S\) be a nonempty subset of vertex set $V$, and let  $\epsilon$ be an error parameter in $(0,1)$. Let  \(\overline{\matx}\) and \(\overline{\maty}\)   denote \(\matq\matw^{1/2}\matb\) and \(\matr\matz^{1/2}\) respectively.   Let \(\matx'\in\rea^{q\times n}\) and \(\maty'\in\rea^{r\times n}\) be matrices such that \(\matx'_{[i,:]}=\mathtt{Solver}(\lap_{-S},\overline{\matx}_{[i,:]},\delta_2)\), \(\maty'_{[i,:]}=\mathtt{Solver}(\lap_{-S},\overline{\maty}_{[i,:]},\delta_2)\), where \(q=r=\ceil{24(\epsilon/7)^{-1}}\log n\), \(\delta_2=\frac{w_{\min}\epsilon}{31n^2}\sqrt{\frac{1-\epsilon/7}{2w_{\max}(1+\epsilon/7)}}\). Then, for any vertex \(i\in V\setminus S\),
    \begin{equation}\label{eq:approx-denom}
        \vece_i^\top\lap_{-S}^{-1}\vece_i = \Abs{\matx\vece_i}^2+\Abs{\maty\vece_i}^2 \approx_{\epsilon/3}\Abs{\matx'\vece_i}^2+\Abs{\maty'\vece_i}^2.
    \end{equation}
\end{lemma}
\begin{proof}
    According to triangle inequality, we have
    \begin{align*}
        \abs{\Abs{\tilde{\matx}\vece_i}-\Abs{\matx'\vece_i}}
         & \leq\Abs{(\tilde{\matx}-\matx)\vece_i}
        \leq\Abs{\tilde{\matx}-\matx'}_{F}                                                                                                                  \\
         & =\sqrt{\sum_{j=1}^q\Abs{\tilde{\matx}_{[j,:]}-\matx'_{[j,:]}}^2}
        \leq\sqrt{\sum_{j=1}^qnw_{\min}^{-1}\Abs{\tilde{\matx}_{[j,:]}-\matx'_{[j,:]}}_{\lap_{-S}}^2}                                                       \\
         & \leq\sqrt{\sum_{j=1}^qnw_{\min}^{-1}\delta_2^2\Abs{\tilde{\matx}_{[j,:]}}_{\lap_{-S}}^2}\leq n\delta_2\sqrt{w_{\min}^{-1}}\Abs{\tilde{\matx}}_F,
    \end{align*}
    where the third inequality and the fourth inequality are due to \lemref{lem:norm-ineq} and \lemref{lemma:ST}, respectively.

    Considering \(q=r=\ceil{24(\epsilon/7)^{-1}}\log n\) and~\lemref{lemma:JL}, we further obtain
    \begin{align*}
        n\delta_2\sqrt{w_{\min}^{-1}}\Abs{\tilde{\matx}}_F
         & \leq n\delta_2\sqrt{w_{\min}^{-1}(1+\epsilon/7)\sum_{i\in V\setminus S}\Abs{\matx\vece_i}^2}                            \\
         & \leq n\delta_2\sqrt{w_{\min}^{-1}(1+\epsilon/7)\trace{\lap_{-S}^{-1}}}\leq n^2\delta_2w_{\min}^{-1}\sqrt{1+\epsilon/7},
    \end{align*}
    where the last inequality is due to \lemref{lem:trace-lap}.
    Similarly, for \(\abs{\Abs{\tilde{\maty}\vece_i}-\Abs{\maty'\vece_i}}\), we have
    \begin{equation}\label{eq:approx-maty-raw}
        \abs{\Abs{\tilde{\maty}\vece_i}-\Abs{\maty'\vece_i}}\le n^2\delta_2w_{\min}^{-1}\sqrt{1+\epsilon/7}.
    \end{equation}
    On the other hand, according to \lemref{lemma:JL}, we have
    \[\Abs{\tilde{\matx}\vece_i}^2+\Abs{\tilde{\maty}\vece_i}^2\ge(1-\epsilon/7)\vece_i^\top\lap_{-S}^{-1}\vece_i\ge(1-\epsilon/7)n^{-2}w_{\max}^{-1}\]
    where the last inequality is due to the scaling of \(\vece_i^\top\lap_{-S}^{-1}\vecd_{-S}\ge1\).

    According to the Pigeonhole Principle, at least one element of \(\setof{\Abs{\tilde{\matx}\vece_i},\Abs{\tilde{\maty}\vece_i}}\) is no less than \(\frac{1}{n}\sqrt{\frac{1-\epsilon/7}{2w_{\max}}}\).
    Without loss of generality, we suppose \(\Abs{\tilde{\matx}\vece_i}\ge\frac{1}{n}\sqrt{\frac{1-\epsilon/7}{2w_{\max}}}\). Then, we have
    \[\frac{\abs{\Abs{\tilde{\matx}\vece_i}-\Abs{\matx'\vece_i}}}{\Abs{\tilde{\matx}\vece_i}}\le n\delta_2w_{\min}^{-1}\sqrt{\frac{2w_{\max}(1+\epsilon/7)}{1-\epsilon/7}}=\frac{\epsilon}{31}.\]
    In addition, we obtain
    \begin{equation*}
        \frac{\abs{\Abs{\tilde{\matx}\vece_i}^2-\Abs{\matx'\vece_i}^2}}{\Abs{\tilde{\matx}\vece_i}^2}
        =  \frac{\abs{\Abs{\tilde{\matx}\vece_i}-\Abs{\matx'\vece_i}}\mypar{\Abs{\tilde{\matx}\vece_i}+\Abs{\matx'\vece_i}}}{\Abs{\tilde{\matx}\vece_i}^2}
        \leq\frac{\epsilon}{31}\mypar{2+\frac{\epsilon}{31}}\le\frac{\epsilon}{15},
    \end{equation*}
    which means
    \begin{equation}\label{eq:approx-matx}
        \frac{\abs{\Abs{\tilde{\matx}\vece_i}^2-\Abs{\matx'\vece_i}^2}}{\Abs{\tilde{\matx}\vece_i}^2+\Abs{\tilde{\maty}\vece_i}^2}\le\frac{\abs{\Abs{\tilde{\matx}\vece_i}^2-\Abs{\matx'\vece_i}^2}}{\Abs{\tilde{\matx}\vece_i}^2}\le\frac{\epsilon}{15}.
    \end{equation}
    Using~\eqref{eq:approx-maty-raw}, it is not difficult to  obtain
    \[\frac{\abs{\Abs{\tilde{\maty}\vece_i}-\Abs{\maty'\vece_i}}}{\Abs{\tilde{\matx}\vece_i}}\le n\delta_2w_{\min}^{-1}\sqrt{\frac{2w_{\max}(1+\epsilon/7)}{1-\epsilon/7}}=\frac{\epsilon}{31},\]
    which indicates
    \begin{equation}\label{eq:approx-maty}
        \frac{\abs{\Abs{\tilde{\maty}\vece_i}^2-\Abs{\maty'\vece_i}^2}}{\Abs{\tilde{\matx}\vece_i}^2+\Abs{\tilde{\maty}\vece_i}^2}
        \le \frac{\abs{\Abs{\tilde{\maty}\vece_i}-\Abs{\maty'\vece_i}}\mypar{\Abs{\tilde{\maty}\vece_i}+\Abs{\maty'\vece_i}}}{\Abs{\tilde{\matx}\vece_i}^2}
        \le \frac{\epsilon}{31}\mypar{2+\frac{\epsilon}{31}}\le\frac{\epsilon}{15}.
    \end{equation}
    Combining~\eqref{eq:approx-matx} and~\eqref{eq:approx-maty}, we are finally able to get
    \begin{equation*}
        \frac{\abs{\mypar{\Abs{\matx'\vece_i}^2+\Abs{\maty'\vece_i}^2}-\mypar{\Abs{\tilde{\matx}\vece_i}^2+\Abs{\tilde{\maty}\vece_i}^2}}}{\Abs{\matx'\vece_i}^2+\Abs{\maty'\vece_i}^2}
        \le \frac{\abs{\Abs{\tilde{\matx}\vece_i}^2-\Abs{\matx'\vece_i}^2}+\abs{\Abs{\tilde{\maty}\vece_i}^2-\Abs{\maty'\vece_i}^2}}{\Abs{\matx'\vece_i}^2+\Abs{\maty'\vece_i}^2}\le\frac{\epsilon}{7},
    \end{equation*}
    which, together with~\lemref{lemma:JL}, completes the proof.
\end{proof}

\subsection{Efficient Evaluation of Marginal Gain \(\Delta(u,S)\)}

Combining Lemmas~\ref{lem:approx-numer} and~\ref{lem:approx-denom}, we propose an efficient algorithm called \textsc{ApproxDelta} approximating \(\Delta(u,S)\) for  \(S\neq\varnothing\). The outline of  \textsc{ApproxDelta} is presented in Algorithm~\ref{algo:approxdelta}. The \textsc{ApproxDelta} algorithm calls the \(\mathtt{Solver}\) in \lemref{lem:approx-denom} for \(q=\ceil{24(\epsilon/7)^{-2}\log n}\) times, resulting in a complexity of \(\tilde{O}\mypar{m\epsilon^{-2}}\). The approximation guarantee for Algorithm~\ref{algo:approxdelta} is provided in
Lemma~\ref{lem:approx-marginest}.

\begin{algorithm}
    \caption{\textsc{ApproxDelta}\((\gr,S,\epsilon)\)}
    \label{algo:approxdelta}
    \Input{
        \(\gr=(V,E,w)\): a connected weighted undirected graph;
        \(S\subseteq V\): the absorbing vertex set;
        \(\epsilon\in(0,1)\): an approximation parameter
    }
    \Output{
        \(\Delta(u,S)\): the margin for vertex \(u\in V\setminus S\)
    }
    \(\lap=\) Laplacian of \(\gr\), \(d_u=\) the degree of \(u\) for all \(u\in V\), \(d=\sum_{u \in V} d_u\)\;
    \(n=\abs{V},m=\abs{E}\)\;
    \(w_{\min}=\min\setof{w_e|e\in E},w_{\max}=\max\setof{w_e|e\in E}\)\;
    \(\delta_1= \frac{w_{\min}\epsilon}{7n^2w_{\max}}\),\(\delta_2= \frac{w_{\min}\epsilon}{31n^2}\sqrt{\frac{1-\epsilon/7}{2w_{\max}(1+\epsilon/7)}}\)\;
    \(\vecx'=\mathtt{Solver}(\lap_{-S},\vecd_{-S},\delta_1)\)\;
    Construct matrix \(\matq_{q\times m}\) and \(\matr_{q\times n}\), where \(q=\ceil{24(\epsilon/7)^{-2}\log n}\) and each entry is \(\pm 1/\sqrt{q}\) with identical probability\;
    Decompose \(\lap_{-S}\) into \(\matb'\in\rea^{m\times n}\), \(\matw'\in\rea^{m\times m}\) and \(\matz\in\rea^{n\times n}\), satisfying \(\lap_{-S}=\matb'^\top\matw'\matb'+\matz\)\;
    \(\overline{\matx}=\matq\matw^{1/2}\matb\),\(\overline{\maty}=\matr\matz^{1/2}\)\;
    \For{\(i=1\) to \(q\)}{
    \(\matx'_{[i,:]}=\mathtt{Solver}(\lap_{-S},\overline{\matx}_{[i,:]},\delta_2)\)\;
    \(\maty'_{[i,:]}=\mathtt{Solver}(\lap_{-S},\overline{\maty}_{[i,:]},\delta_2)\)\;
    }
    \ForEach{\(u\in V\setminus S\)}{
        \(\Delta'(u,S)=\frac{\vecx'^2}{d\mypar{\Abs{\matx'\vece_i}^2+\Abs{\maty'\vece_i}^2}}\)\;
    }
    \Return{\(\setof{\Delta'(u,S)|u\in V\setminus S}\)}

\end{algorithm}

\begin{lemma}\label{lem:approx-marginest}
    For any real number \(\epsilon\in(0,1)\), the returned value \(\Delta'(u,S)\) of Algorithm~\ref{algo:approxdelta} satisfies
    \begin{equation*}
        \Delta(u,S)\approx_\epsilon \Delta'(u,S)
    \end{equation*}
    with high probability.
\end{lemma}

\subsection{Nearly Linear Algorithm for the MinGWC Problem}

By combining \algoref{ALG01} and \algoref{algo:approxdelta}, we develop a faster greedy approximation algorithm called \(\textsc{ApproxMinGWC}\), outlined in Algorithm~\ref{algo:approx}. The algorithm has time  complexity \(\tilde{O}\mypar{km\epsilon^{-2}}\), which  is nearly linear in $m$,  since it calls \textsc{ApproxHK} once and \textsc{ApproxDelta} \(k-1\) times. Algorithm~\ref{algo:approx} has a $(1-\frac{k}{k-1}\frac{1}{e}-\epsilon)$ approximation factor, as shown in Theorem~\ref{Mainthm}.

\begin{algorithm}
    \caption{\textsc{ApproxMinGWC}\((\gr,k,\epsilon)\)}
    \label{algo:approx}
    \Input{
        \(\gr=(V,E,w)\): a connected weighted undirected graph;
        \(k\ll\abs{V}\): the capacity of vertex set;
        \(\epsilon\in(0,1)\): an approximation parameter
    }
    \Output{\(S_k\): A subset of \(V\) with \(\abs{S_k}=k\)}
    \(\setof{\tilde{H}_u|u\in V }=\textsc{ApproxHK}(\gr,\epsilon)\)\;
    \(S_1=\setof{\argmin_{u\in V}\setof{\tilde{H}_u}}\)\;
    \For{\(i=2\) to \(k\)}{
        \(\setof{\Delta'(u,S)|u\in V\setminus S}=\textsc{ApproxDelta}(\gr,S,\epsilon)\)\;
        \(u^*=\argmax_{u\in V\setminus S}\setof{\Delta'(u,S)}\)\;
        \(S_i= S_{i-1}\cup\setof{u^*}\)\;
    }
    \Return{\(S_k\)}
\end{algorithm}

\begin{theorem}\label{Mainthm}
    The algorithm \(\textsc{ApproxMinGWC}(\gr,k,\epsilon)\) takes a connected weighted undirected graph \(\gr=(V,E,w)\) with Laplacian matrix \(\lap\) , a positive integer \(1<k \ll n\), and an error parameter \(\epsilon\in(0,1)\), and returns a set  \(S_k\) of \(k\) vertices. With high probability, the returned vertex set \(S\) satisfies
    \begin{equation*}
        \mypar{1+\epsilon}\manc{\setof{u^*}}-\manc{S_k}\geq\mypar{1-\frac{k}{k-1}\frac{1}{e}-\epsilon}\mypar{\mypar{1+\epsilon}\manc{\setof{u^*}}-\manc{S^*}},
    \end{equation*}
    where
    \[u^*\defeq\argmin_{u\in V}\manc{\setof{u}}, S^*\defeq\argmin_{\abs{S}=k}\manc{S}.\]
\end{theorem}
\begin{proof}
    Considering~\lemref{lem:approx-marginest} and the supermodular property of the objective function $H(\cdot)$, we are able to obtain  that for any positive integer \(i\),
    \begin{equation*}
        \manc{S_i}-\manc{S_{i+1}}\ge\frac{1-\epsilon}{k}\mypar{\manc{S_i}-\manc{S^*}},
    \end{equation*}
    which indicates
    \begin{equation*}
        \manc{S_{i+1}}-\manc{S^*}\le\bigpar{1-\frac{1-\epsilon}{k}}\mypar{\manc{S_i}-\manc{S^*}}.
    \end{equation*}
    Repeating the above process, we further obtain
    \begin{equation*}
        \manc{S_k}-\manc{S^*} \le \bigpar{1-\frac{1-\epsilon}{k}}^{k-1}\mypar{\manc{S_1}-\manc{S^*}}
        \le                       \bigpar{\frac{k}{k-1}\frac{1}{e}+\epsilon}\mypar{\manc{S_1}-\manc{S^*}},
    \end{equation*}
    which, together with the fact that \(\manc{S_1}\le\mypar{1+\epsilon}\manc{\setof{u^*}}\),
    completes the proof.
\end{proof}

\section{Numerical Experiments}

In this section, we present experimental results for various real and model networks to demonstrate the  efficiency and accuracy of our  approximation algorithms:  \textsc{ApproxHK} for computing the walk centrality \(H_u\) of  every vertex $u$ and the Kemeny constant \(\kem\), \textsc{DeterMinGWC} and \textsc{ApproxMinGWC} for solving the problem of Minimizing GWC.

\subsection{Experimental Settings}

\textbf{Datasets.} The data of realistic networks are taken from the Koblenz Network Collection~\cite{Ku13} and SNAP~\cite{LeKr14}. For those networks that are disconnected originally, we perform our experiments on their largest connected components (LCCs). Related information about the number of vertices and edges for the studied real networks and their LCCs is shown  in \tabref{tab:runtime_comparison},  where networks are listed in an increasing order of the number of vertices in the networks. The smallest network has 198 vertices, while the largest one consists of more than \(1.9 \times 10^{6}\) vertices.

\textbf{Environment.} All experiments are conducted on a machine with 16-core 2.55GHz AMD EPYC Milan CPU and with 64GB of RAM. We implement all the  algorithms in \textit{Julia v1.8.5}. For  \textsc{ApproxHK}  and  \textsc{ApproxMinGWC}, we use the  \(\mathtt{Solver}\) from~\cite{KySa16}.

\subsection{Performance for Algorithm \textsc{ApproxHK} }\label{subsec:perf-approxhk}

In this subsection, we present experimental results for \textsc{ApproxHK} on real and model networks.

\subsubsection{Results on Real Networks}

We first evaluate the efficiency of algorithm \textsc{ApproxHK}  in approximating centrality \(H_u\) for realistic networks.
In \tabref{tab:runtime_comparison}, we report the running time of \textsc{ApproxHK} and that of the direct deterministic algorithm called \textsc{DeterHK} that calculates the centrality \(H_u\) for each \(u \in V\) by calculating the pseudoinverse \(\mypar{\lap +\frac{1}{n}\matj}^{-1} - \frac{1}{n}\matj\) of \(\lap\) as given  in \lemref{HjK}. To objectively evaluate the running time, for both \textsc{DeterHK} and \textsc{ApproxHK} on all considered networks except the last ten ones marked with \(\ast\), we enforce the program to run on 16 threads. From \tabref{tab:runtime_comparison} we can see that for moderate approximation parameter \(\epsilon\), the computational time for \textsc{ApproxHK} is significantly  smaller than  that  for \textsc{DeterHK}, especially for large-scale networks tested.  Thus, \textsc{ApproxHK} can significantly improve the performance compared with \textsc{DeterHK}. For the  ten networks marked with \(\ast\), the number of vertices for which ranges from \(10^5\) to \(10^6\), we cannot run the \textsc{DeterHK} algorithm on the machine due to the limits of memory and time. However, for these networks, we can approximately compute their walk centrality for all vertices by using algorithm \textsc{ApproxHK}, which further shows that \textsc{ApproxHK} is efficient and scalable to large networks.

\begin{table*}[htbp]
    \centering
    \normalsize
    \fontsize{8.0}{8.8}\selectfont
    \begin{threeparttable}
        \caption{The running time (seconds, \(s\)) of \textsc{DeterHK} and \textsc{ApproxHK} with various \(\epsilon\) on a diverse range of realistic networks from various domains}
        \label{tab:runtime_comparison}
        \begin{tabular}{crrcccccc}
            \toprule
            \multirow{2}{*}{Network}          &
            \multirow{2}{*}{Vertices}         &
            \multirow{2}{*}{Edges}            &
            \multirow{2}{*}{\textsc{DeterHK}} &
            \multicolumn{5}{c}{\textsc{ApproxHK} (\(s\)) with various \(\epsilon\)}                                                 \\
            \cmidrule{5-9}                    &           &           & (\(s\)) & \(0.3\) & \(0.25\) & \(0.2\) & \(0.15\) & \(0.1\) \\
            \midrule
            Jazz musicians                    & 198       & 2,742     & 0.096   & 0.006   & 0.006    & 0.029   & 0.100    & 0.205   \\
            Euroroads                         & 1,039     & 1,305     & 0.044   & 0.065   & 0.095    & 0.115   & 0.124    & 0.130   \\
            Hamster friends                   & 1,788     & 12,476    & 0.264   & 0.094   & 0.120    & 0.124   & 0.133    & 0.193   \\
            Hamster full                      & 2,000     & 16,098    & 0.184   & 0.063   & 0.091    & 0.129   & 0.146    & 0.586   \\
            Facebook (NIPS)                   & 4,039     & 88,234    & 1.386   & 0.157   & 0.232    & 0.341   & 0.546    & 1.130   \\
            CA-GrQc                           & 4,158     & 13,422    & 1.260   & 0.086   & 0.150    & 0.160   & 0.378    & 0.456   \\
            US power grid                     & 4,941     & 6,594     & 2.012   & 0.040   & 0.074    & 0.109   & 0.123    & 0.221   \\
            Reactome                          & 5,973     & 145,778   & 3.447   & 0.273   & 0.377    & 0.568   & 0.854    & 1.968   \\
            Route views                       & 6,474     & 12,572    & 4.000   & 0.075   & 0.108    & 0.111   & 0.260    & 0.272   \\
            CA-HepTh                          & 8,638     & 24,827    & 9.248   & 0.105   & 0.174    & 0.224   & 0.318    & 0.467   \\
            Sister cities                     & 10,320    & 17,988    & 14.75   & 0.109   & 0.248    & 0.320   & 0.992    & 1.357   \\
            Pretty Good Privacy               & 10,680    & 24,316    & 16.40   & 0.360   & 0.381    & 0.534   & 0.640    & 1.156   \\
            CA-HepPh                          & 11,204    & 117,619   & 18.72   & 0.374   & 0.495    & 0.873   & 1.730    & 3.067   \\
            Astro-ph                          & 17,903    & 196,972   & 74.08   & 0.873   & 1.114    & 1.976   & 2.912    & 6.592   \\
            CAIDA                             & 26,475    & 53,381    & 233.0   & 0.353   & 0.964    & 1.184   & 1.381    & 2.330   \\
            Brightkite*                       & 56,739    & 212,945   & --      & 1.346   & 2.462    & 3.406   & 4.784    & 16.21   \\
            Livemocha*                        & 104,103   & 2,193,083 & --      & 12.64   & 16.32    & 27.40   & 39.98    & 82.02   \\
            WordNet*                          & 145,145   & 656,230   & --      & 6.905   & 8.072    & 13.89   & 19.47    & 40.63   \\
            Gowalla*                          & 196,591   & 950,327   & --      & 8.150   & 11.05    & 17.03   & 28.93    & 66.72   \\
            com-DBLP*                         & 317,080   & 1,049,866 & --      & 13.97   & 20.55    & 36.30   & 67.26    & 163.0   \\
            Amazon*                           & 334,863   & 925,872   & --      & 20.61   & 31.01    & 64.83   & 99.94    & 223.7   \\
            roadNet-PA*                       & 1,087,562 & 1,541,514 & --      & 98.88   & 158.7    & 270.9   & 416.1    & 1003    \\
            YouTube*                          & 1,134,890 & 2,987,624 & --      & 70.82   & 106.7    & 164.2   & 254.5    & 629.2   \\
            roadNet-TX*                       & 1,351,137 & 1,879,201 & --      & 159.3   & 223.6    & 371.5   & 689.4    & 1652    \\
            roadNet-CA*                       & 1,965,027 & 2,760,388 & --      & 243.7   & 367.5    & 590.9   & 959.3    & 2327    \\
            \bottomrule
        \end{tabular}
    \end{threeparttable}
\end{table*}

In addition to the high efficiency,  our algorithm  \textsc{ApproxHK} also provides a desirable approximation \(\hat{H}_u\) for the walk centrality \(H_u\).   To show the accuracy of  \textsc{ApproxHK}, we compare the  approximate  results for \textsc{ApproxHK} with the exact results computed by  Lemma~\ref{HjK}. In \tabref{tab:accuracy}, we report the mean relative error \(\sigma\) of algorithm  \textsc{ApproxHK}, with \(\sigma\) being defined by \(\sigma=\frac{1}{n}\sum_{u\in V}|{H_u}-\tilde{H}_u|/{H_u}\). From  \tabref{tab:accuracy} we can see that  the actual mean relative errors for all \(\epsilon\) and all networks are  very small, and are almost negligible for smaller \(\epsilon\). More interestingly, for all networks tested,   \(\sigma\) are magnitudes smaller than the theoretical guarantee. Therefore, the  approximation algorithm  \textsc{ApproxHK} provides very accurate results in practice.

\begin{table}[htbp]
    \tabcolsep=8pt
    \centering
    \fontsize{8.0}{8.8}\selectfont
    \begin{threeparttable}
        \caption{Mean relative error \(\sigma\) of \textsc{ApproxHK} (\(\times \num{e-2}\))} 
        \label{tab:accuracy}
        \begin{tabularx}{8.5cm}{c p{0.5cm} p{0.5cm} p{0.5cm} p{0.5cm} p{0.5cm}}
            \toprule[1pt]
            \multirow{2}{*}{Network} &
            \multicolumn{5}{c}{Mean relative error \(\sigma\) (\(\times 10^{-2}\)) for various \(\epsilon\)} \\
            \cmidrule{2-6}
                                     & \(0.3\)   & \(0.25\)  & \(0.2\)   & \(0.15\)  & \(0.1\)               \\
            \midrule
            Jazz musicians           & \(1.933\) & \(1.434\) & \(0.899\) & \(0.444\) & \(0.013\)             \\
            Euroroads                & \(5.718\) & \(3.231\) & \(2.042\) & \(1.455\) & \(0.181\)             \\
            Hamster friends          & \(1.520\) & \(0.385\) & \(0.372\) & \(0.219\) & \(0.068\)             \\
            Hamster full             & \(0.889\) & \(0.256\) & \(0.100\) & \(0.064\) & \(0.012\)             \\
            Facebook (NIPS)          & \(3.407\) & \(2.413\) & \(1.750\) & \(0.420\) & \(0.348\)             \\
            CA-GrQc                  & \(1.656\) & \(1.272\) & \(0.827\) & \(0.221\) & \(0.051\)             \\
            US power grid            & \(3.385\) & \(1.626\) & \(1.166\) & \(0.406\) & \(0.137\)             \\
            Reactome                 & \(0.268\) & \(0.502\) & \(0.146\) & \(0.054\) & \(0.048\)             \\
            Route views              & \(0.123\) & \(0.076\) & \(0.069\) & \(0.060\) & \(0.044\)             \\
            CA-HepTh                 & \(0.589\) & \(0.269\) & \(0.215\) & \(0.180\) & \(0.056\)             \\
            Sister cities            & \(0.190\) & \(0.140\) & \(0.124\) & \(0.086\) & \(0.018\)             \\
            Pretty Good Privacy      & \(0.534\) & \(0.491\) & \(0.376\) & \(0.173\) & \(0.084\)             \\
            CA-HepPh                 & \(0.916\) & \(0.406\) & \(0.235\) & \(0.135\) & \(0.132\)             \\
            Astro-ph                 & \(0.341\) & \(0.335\) & \(0.117\) & \(0.016\) & \(0.009\)             \\
            CAIDA                    & \(0.100\) & \(0.056\) & \(0.049\) & \(0.010\) & \(0.003\)             \\
            \bottomrule
        \end{tabularx}
    \end{threeparttable}
\end{table}

\subsubsection{Results on Model Networks}

To further demonstrate the performance of our proposed algorithm \textsc{ApproxHK}, we use it to compute the Kemeny constant for some model networks.  Although for a general graph, the exact result for its  Kemeny constant is difficult to obtain, for some model networks generated by an iterative way, their Kemeny constant can be determined explicitly.
For example, for the pseudofractal scale-free web~\cite{XiZhCo16,ShLiZh17}, the Koch network~\cite{XiLiZh15}, the Cayley tree~\cite{CaCh97,ChCa99}, and the extended Tower of Hanoi graph~\cite{KlMo05}, one can obtain exact expressions for their Kemeny constant.

In the following we use algorithm \textsc{ApproxHK} to approximately compute the Kemeny constant for the aforementioned four model networks. One main justification for selecting  these four networks lies in their relevance to real-world systems. For example,  the pseudofractal scale-free web and the Koch network display the remarkable scale-free small-world properties as observed in most real networks~\cite{Ne03}, while the Cayley tree is useful for parallel and distributed computing, since it can be embedded into many other interconnection architectures, such as butterfly networks~\cite{GuHa91} and M{\"o}bius cubes~\cite{LiFaJi16}.

\textbf{Kemeny Constant of Pseudofractal Scale-Free Web.} Let \(\mathcal{F}_g\) (\(g \geq 0\)) denote the pseudofractal scale-free web after \(g\) iterations. For \(g=0\), \( \mathcal{F}_0\) consists of a triangle of three vertices and three edges. For \(g>0\), \(\mathcal{F}_g\) is obtained from \(\mathcal{F}_{g-1}\) as follows. For each existing edge in \(\mathcal{F}_{g-1}\), a new vertex is created and linked to both end vertices of the edge.  \figref{psfw1} schematically illustrates the construction process of the pseudofractal scale-free web. In network \(\mathcal{F}_g\), there are \((3^{g+1}+3)/2\) vertices and \(3^{g+1}\) edges.  It has been shown~\cite{XiZhCo16} that the Kemeny constant \(\kem(\mathcal{F}_g) \) for \(\mathcal{F}_g\) is
\begin{equation}\label{Kg01}
    \kem(\mathcal{F}_g)=\frac{5}{2}\times3^g-\frac{5}{3}\times2^g+\frac{1}{2}\,. 
\end{equation}

\begin{figure}[!t]
    \centering
    \begin{minipage}{0.475\linewidth}
        \begin{minipage}{\linewidth}
            \centering
            \includegraphics[width=0.85\linewidth]{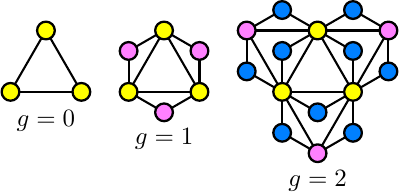}
            \caption{The first several iterations of the pseudofractal scale-free web.}
            \Description{The first several iterations of the pseudofractal scale-free web.}
            \label{psfw1}
        \end{minipage}
        \\
        \begin{minipage}{\linewidth}
            \centering
            \includegraphics[width=0.85\linewidth]{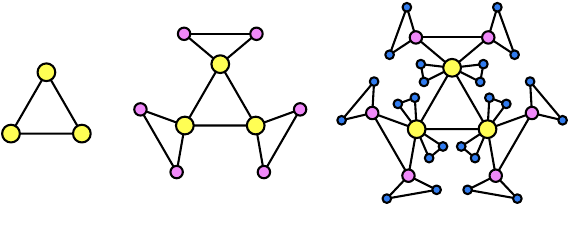}
            \caption{Construction process for the Koch network.}
            \Description{Construction process for the Koch network.}
            \label{network}
        \end{minipage}
    \end{minipage}
    \hfill
    \begin{minipage}{0.475\linewidth}
        \centering
        \includegraphics[width=0.975\linewidth]{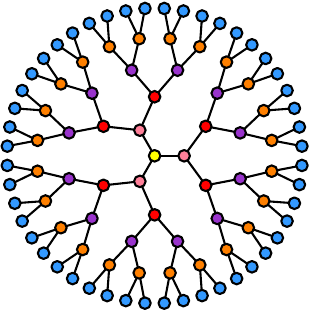}
        \caption{The Cayley tree \(\mathcal{C}_{3,5}\).}
        \Description{The Cayley tree.}
        \label{Cayley}
    \end{minipage}
\end{figure}

\textbf{Kemeny Constant of Koch Network.}
The Koch network is also built in an iterative way~\cite{ZhZhXiLiGu09}. Let \(\mathcal{M}_{g}\) (\(g \geq 0\)) denote the Koch network after \(g\) iterations. Initially (\(g=0\)), \(\mathcal{M}_{0}\) is a triangle with three
vertices and three edges. For \(g\geq 1\), \(\mathcal{M}_{g}\) is obtained from \(\mathcal{M}_{g-1}\) by
performing the following operations. For each of the three vertices in
every existing triangle in \(\mathcal{M}_{g-1}\), two new vertices are created, both of which and their
``mother'' vertices are connected to one another forming a new triangle. \figref{network} illustrates the growth process of the Koch network.  In network \(\mathcal{M}_{g}\), the number of vertices is \(2\times 4^{g}+1\), and the number of
edges is \(3\times 4^{g}\).  In~\cite{XiLiZh15}, the Kemeny constant \(\kem(\mathcal{M}_g)\) for \(\mathcal{M}_g\) was obtained to be
\begin{equation}\label{Kg02}
    \kem(\mathcal{M}_g)=(1+2g)\times 4^g+\frac{1}{3}\,. 
\end{equation}

\textbf{Kemeny Constant of Cayley Tree.}
Let \(\mathcal{C}_{b,g}(b\ge3,g\ge0)\) represent the Cayley tree after
\(g\) iterations, which are constructed  as follows~\cite{CaCh97,ChCa99}. Initially \((g = 0)\), \(\mathcal{C}_{b,0}\) only contains a central vertex. To obtain \(\mathcal{C}_{b,1}\), we generate  \(b\) vertices and linked them to
the central one. For any \(g>1\), \(\mathcal{C}_{b,g}\) is obtained from \(\mathcal{C}_{b,g-1}\) by performing the following operation. For every periphery vertex of \(\mathcal{C}_{b,g-1}\), \(b-1\) vertices are created and connected to
the periphery vertex. \figref{Cayley} illustrates a special Cayley tree,
\(\mathcal{C}_{3,5}\). In the tree \(\mathcal{C}_{b,g}\), there are \((b(b-1)^g-2)/(b-2)\) vertices and \((b(b-1)^g-2)/(b-2)-1\) edges. In~\cite{JuWuZh13}, the authors  considered a particular Cayley tree \(\mathcal{C}_{3,g}\) corresponding to \(b=3\) and obtained an exact expression for its Kemeny constant \(\kem(\mathcal{C}_{3,g})\):
\begin{equation}\label{Kg03}
    \kem(\mathcal{C}_{3,g})=\frac{3g\times 4^{g+1}-13\times 2^{2g+1} + 35\times 2^g - 9}{2(2^g-1)}.
\end{equation}

\textbf{Kemeny Constant of Extended Tower of Hanoi graph.}  As the name suggests, the extended Tower of Hanoi graph is constructed  based on the Tower of Hanoi graph.  Let \(\mathcal{H}_{g}\)  be the  Tower of Hanoi graph of generation \(g\)~\cite{HiKlMiPeSt13}. The vertex set of \(\mathcal{H}_{g}\) consists of all \(3\)-tuples of integers \(1,2,3\), i.e., \(V(\mathcal{H}_{g})=\{1,2,3\}^g\). All vertices in \(\mathcal{H}_{g}\) are labelled as \((x_1,x_2,\cdots,x_g)\), hereafter written in regular expression form \(x_1x_2\cdots x_g\), with \(x_i\in\{1,2,3\}\) for \(i=1,2,...,g\). Two vertices \(u=u_1u_2\cdots u_g\) and \(v=v_1v_2\cdots v_g\) in \(\mathcal{H}_{g}\) are directly connected by an edge if and only if there exists an \(h(1\le h\le g)\) such that
\begin{enumerate}
    \item \(u_t=v_t\) for \(1\le t\le h-1\);
    \item \(u_h\neq v_h\);
    \item \(u_t=v_h\) and \(v_t=u_h\) for \(h+1\le t\le g\).
\end{enumerate}

Figure~\ref{Sierpinski} illustrates the  Tower of Hanoi graph \(\mathcal{H}_{3}\) and its vertex labeling. In \(\mathcal{H}_{g}\), a vertex having label form \((ii\cdots i)(1\le i\le 3)\) is called an extreme vertex.

\begin{figure}[!t]
    \centering
    \subfigure[\(\mathcal{H}_{3}\)]{
        \begin{minipage}[t]{0.35\linewidth}
            \centering
            \includegraphics[width=0.95\linewidth,trim=0 0 0 0]{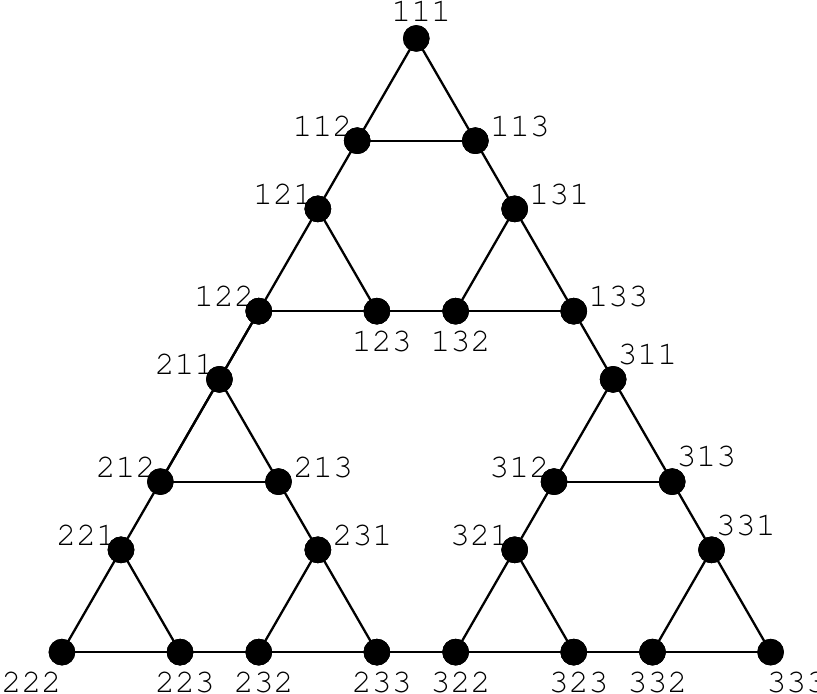}
            \label{Sierpinski}
        \end{minipage}%
    }%
    \subfigure[\(\overline{\mathcal{H}}_{3}\)]{
        \begin{minipage}[t]{0.35\linewidth}
            \centering
            \includegraphics[width=0.95\linewidth,trim=0 0 0 0]{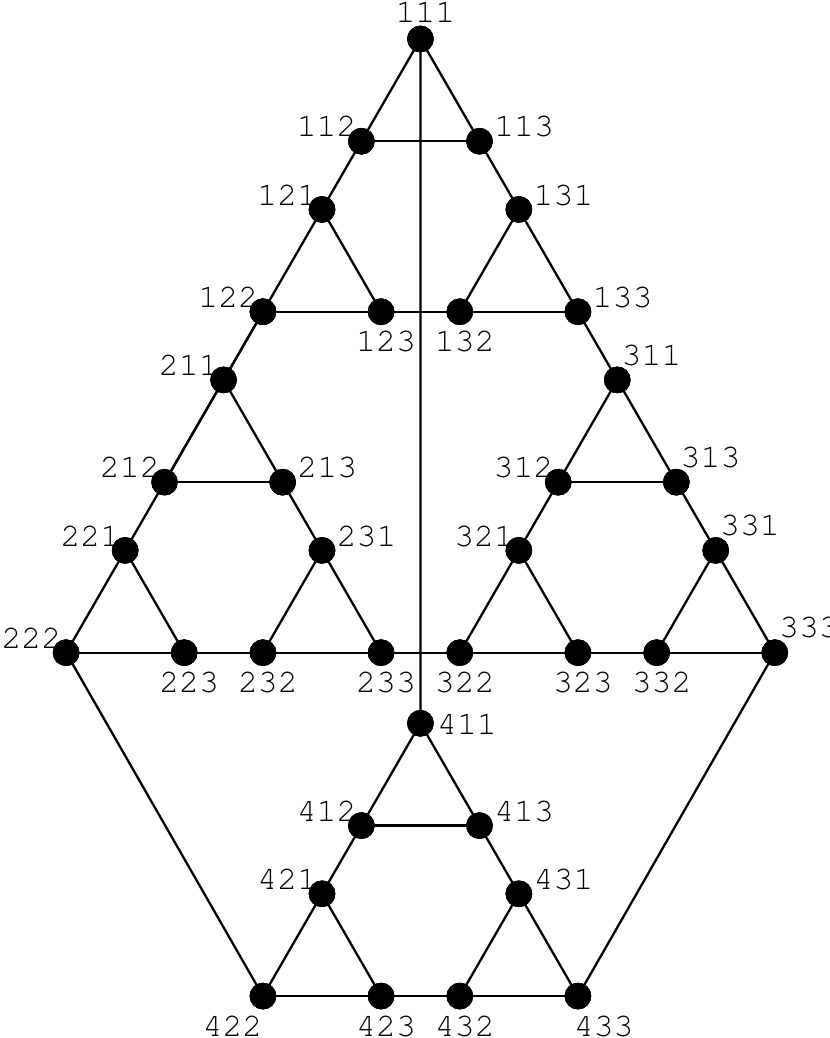}
            \label{exSierpinski}
        \end{minipage}%
    }%
    \caption{Illustrations of the Tower of Hanoi graph \(\mathcal{H}_{3}\) and its extension \(\overline{\mathcal{H}}_{3}\), as well as their vertex labeling.}
    \Description{Illustrations of the Tower of Hanoi graph and its extension, as well as their vertex labeling.}
\end{figure}

The extended  Tower of Hanoi graph denoted by \(\overline{\mathcal{H}}_{g}(g\ge 1)\)  proposed by Klava\u zar and Mohar~\cite{KlMo05,ZhWuLiCo16,QiDoZhZh20}, is defined as follows. For \(g=1\), \(\overline{\mathcal{H}}_{1}\) is a triangle. For \(g\ge2\),  \(\overline{\mathcal{H}}_{g}\) was obtained from the disjoint union of \(4\) copies of \(\mathcal{H}_{g-1}\), where the extreme vertices in different  replicas of \(\mathcal{H}_{g-1}\) are connected as the 4-vertex complete graph. Figure~\ref{exSierpinski} shows the extended  Tower of Hanoi graph \(\overline{\mathcal{H}}_{3}\) and the its labeling. In graph \(\overline{\mathcal{H}}_{g}\), the number of nodes is \(4\cdot3^{g-1}\), and the number of edges is \(4\cdot3^g/2\). As shown in~\cite{QiZh18}, the Kemeny constant \(\kem(\overline{\mathcal{H}}_{g})\) for \(\overline{\mathcal{H}}_{g}\) is
\begin{align}
    \kem(\overline{\mathcal{H}}_{g}) = \frac{32\times5^g\times3^{g-1}-64\times3^{2g-2}-2\times3^g}{10(3^g+3^{g-1}-1)}.
    \label{Kg04}
\end{align}

\textbf{Numerical Results.} We use our algorithm \textsc{ApproxHK} to compute the Kemeny constant on pseudofractal scale-free web \(\mathcal{F}_{12}\),  Koch network \(\mathcal{M}_{10}\), Cayley tree \(\mathcal{C}_{3,19}\), and extended  Tower of Hanoi graph \(\overline{\mathcal{H}}_{13}\). The numerical results are reported in Table~\ref{tab:Kemeny}, which shows that the approximation algorithm \textsc{ApproxHK} works effectively for the four networks. This again demonstrates the advantage of our proposed algorithm  \textsc{ApproxHK} for large networks.

\begin{table*}[htbp]
    \centering
    \begin{threeparttable}
        \caption{Exact Kemeny constant \(\kem\), their approximation \(\approxkem\), relative error \(\rho=\abs{K-\tilde{K}}/K\), and running time (seconds, \(s\)) for \(\tilde{K}\) on networks \(\mathcal{F}_{12}\), \(\mathcal{M}_{10}\), \(\mathcal{C}_{3,19}\) and \(\overline{\mathcal{H}}_{13}\).  \(K\) is obtained via~\eqref{Kg01} and~\eqref{Kg02}, while \(\tilde{K}\) is obtained through algorithm \textsc{ApproxHK} with \(\epsilon=0.2\)}
        \label{tab:Kemeny}
        \begin{tabular}{crrrrcc}
            \toprule
            Network                         & \multicolumn{1}{c}{Vertices} & \multicolumn{1}{c}{Edges} & \multicolumn{1}{c}{\(\kem\)} & \multicolumn{1}{c}{\(\approxkem\)} & Error \(\rho\) & Time\cr
            \midrule
            \specialrule{0em}{3pt}{3pt}
            \(\mathcal{F}_{12}\)            & 797,163                      & 1,594,323                 & 1,321,776                    & 1,322,243                          & 0.00035        & 16.71\cr
            \specialrule{0em}{3pt}{3pt}
            \(\mathcal{M}_{10}\)            & 2,097,153                    & 3,145,728                 & 22,020,096                   & 22,015,912                         & 0.00004        & 45.70\cr
            \specialrule{0em}{3pt}{3pt}
            \(\mathcal{C}_{3,19}\)          & 1,572,862                    & 1,572,861                 & 52,953,206                   & 52,564,893                         & 0.00733        & 31.48\cr
            \specialrule{0em}{3pt}{3pt}
            \(\overline{\mathcal{H}}_{13}\) & 2,125,764                    & 3,188,646                 & 975,712,653                  & 970,030,470                        & 0.00582        & 1028\cr
            \specialrule{0em}{3pt}{3pt}
            \bottomrule
        \end{tabular}
    \end{threeparttable}
\end{table*}

\subsection{Performance for Algorithms \textsc{DeterMinGWC} and \textsc{ApproxMinGWC}}\label{subsec:perf-approxmingwc}

\textbf{Accuracy of Algorithms \textsc{DeterMinGWC} and \textsc{ApproxMinGWC}.} We first demonstrate the accuracy of \textsc{DeterMinGWC} and \textsc{ApproxMinGWC} by comparing their results of with the optimum and random solutions on four small-scale networks~\cite{Ku13}: \(\mathit{Zebra}\) with 23 vertices and 105 edges, \(\mathit{Zachary\ karate\ club}\) with 34 vertices and 78 edges, \(\mathit{Contiguous\ USA}\) with 49 vertices and 107 edges, and \(\mathit{Les\ Miserables}\) with 77 vertices and 254 edges. For each network, a random solution is obtained by choosing $k$ vertices uniformly. Moreover, since these are small networks, we are able to compute the optimum solutions by brute-force search. For each $k=2,3,4,5$, we find $k$ vertices by optimal and stochastic schemes, and \textsc{DeterMinGWC} and \textsc{ApproxMinGWC}, with the error parameter $\epsilon$ in \textsc{ApproxMinGWC} being \(0.2\). We report the GWC results on these four strategies in~\figref{pic:compare-effect-optimum}. As shown in~\figref{pic:compare-effect-optimum}, the solutions provided by \textsc{DeterMinGWC} and \textsc{ApproxMinGWC} are almost identical, with both being very close to the optimum solutions. This indicates that the practical approximation ratios of our algorithms are much better than their theoretical guarantees. In addition,  \textsc{DeterMinGWC} and \textsc{ApproxMinGWC} are far better than the random strategy.

\begin{figure}[!t]
    \centering
    \includegraphics[width=0.75\linewidth]{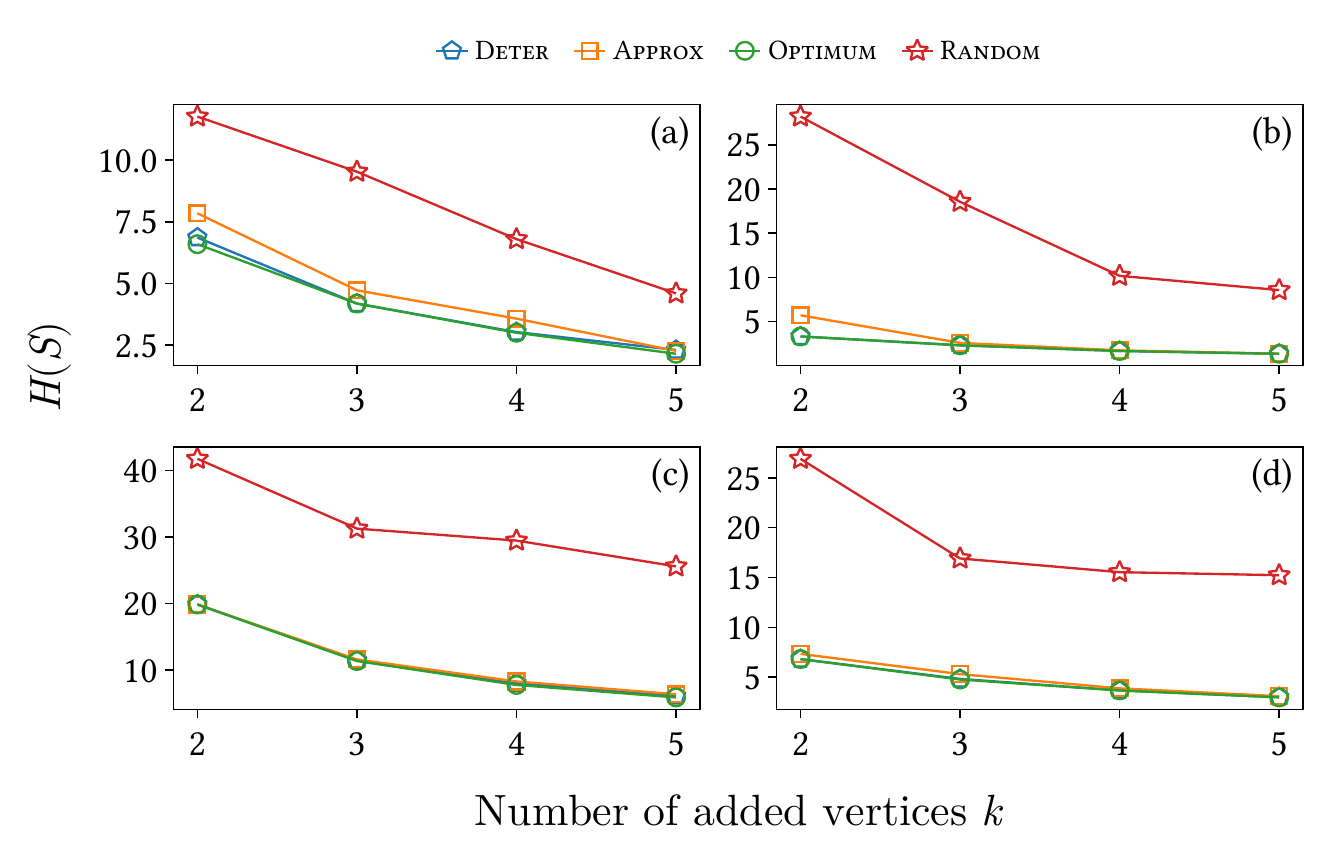}
    \caption{GWC \(\manc{S}\) of vertex group \(S\) computed by four different algorithms, \textsc{DeterMinGWC} (\textsc{Deter}), \textsc{ApproxMinGWC} (\textsc{Approx}), \textsc{Random}, and \textsc{Optimum}, on four networks: Zebra (a), Zachary karate club (b), Contiguous USA (c), and Les Miserables (d).}
    \Description{GWC of vertex group computed by four different algorithms on four networks: Zebra (a), Zachary karate club (b), Contiguous USA (c), and Les Miserables (d).}
    \label{pic:compare-effect-optimum}
\end{figure}

We continue to demonstrate the effectiveness of our algorithms \textsc{DeterMinGWC} and \textsc{ApproxMinGWC} by comparing their performance with four baseline algorithms: \textsc{Between}, \textsc{Top-Absorb}, \textsc{Top-Degree}, and \textsc{Top-PageRank}. \textsc{Between} selects a group of \(k\) vertices with maximum value of group betweenness centrality~\cite{PuElDo07}. Group betweenness centrality is extended from the concept of betweenness centrality, which measures the importance of individual vertices based on shortest paths. We implement Algorithm 2 from~\cite{PuElDo07} as the baseline \textsc{Between}, whose time complexity is \(O(n^3+kn^2)\) for a graph with \(n\) vertices. Meanwhile, \textsc{Top-Absorb} simply selects \(k\) vertices with the lowest value of walk centrality, \textsc{Top-Degree} selects \(k\) vertices with the largest degrees, while \textsc{Top-PageRank} selects \(k\) vertices with the largest PageRank values. We execute these six strategies for vertex selection on eight medium-sized networks, with the parameter $\epsilon$ in \textsc{ApproxMinGWC} being \(0.2\). The results for these schemes are shown in Figures~\ref{pic:small-effect-line} and~\ref{pic:medium-effect-line} for various \(k\).

As displayed in these two figures, both of our algorithms produce similar approximation solutions, outperforming the four baseline strategies. Despite its relative disadvantage compared with \textsc{ApproxMinGWC}, \textsc{Top-PageRank} consistently provides a robust solution for MinGWC among the four baselines. This may be attributed to PageRank's interpretation in random walk models, which potentially offer a similar vertex centrality measure to GWC. In contrast, the solution quality for \textsc{Between} is not satisfactory, highlighting the disparity between random walk models and shortest path models. Furthermore, the inferior performance of \textsc{Top-Absorb} indicates that MinGWC cannot be effectively solved by merely selecting \(k\) vertices with minimum walk centrality.

\begin{figure}[!t]
    \centering
    \includegraphics[width=0.75\linewidth]{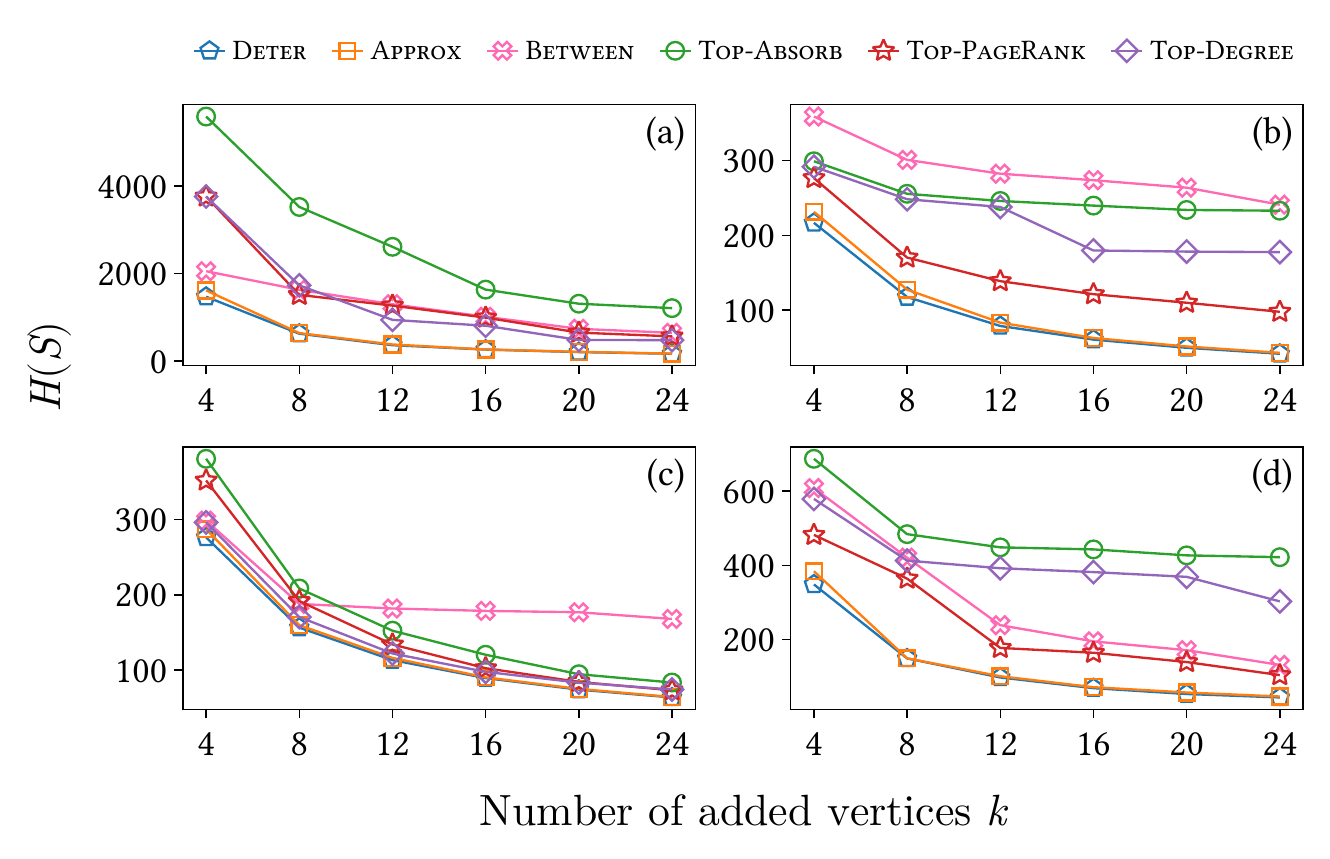}
    \caption{GWC \(\manc{S}\) of vertex group \(S\) computed by six different algorithms, \textsc{DeterMinGWC} (\textsc{Deter}), \textsc{ApproxMinGWC} (\textsc{Approx}), \textsc{Between}, \textsc{Top-Absorb}, \textsc{Top-Degree}, and \textsc{Top-PageRank} on four networks: US power grid (a), CA-GrQc (b), CA-HepTh (c), and Euroroads (d).}
    \Description{GWC of vertex group computed by six different algorithms on four networks: US power grid (a), CA-GrQc (b), CA-HepTh (c), and Euroroads (d).}
    \label{pic:small-effect-line}
\end{figure}

\begin{figure}[!t]
    \centering
    \includegraphics[width=0.75\linewidth]{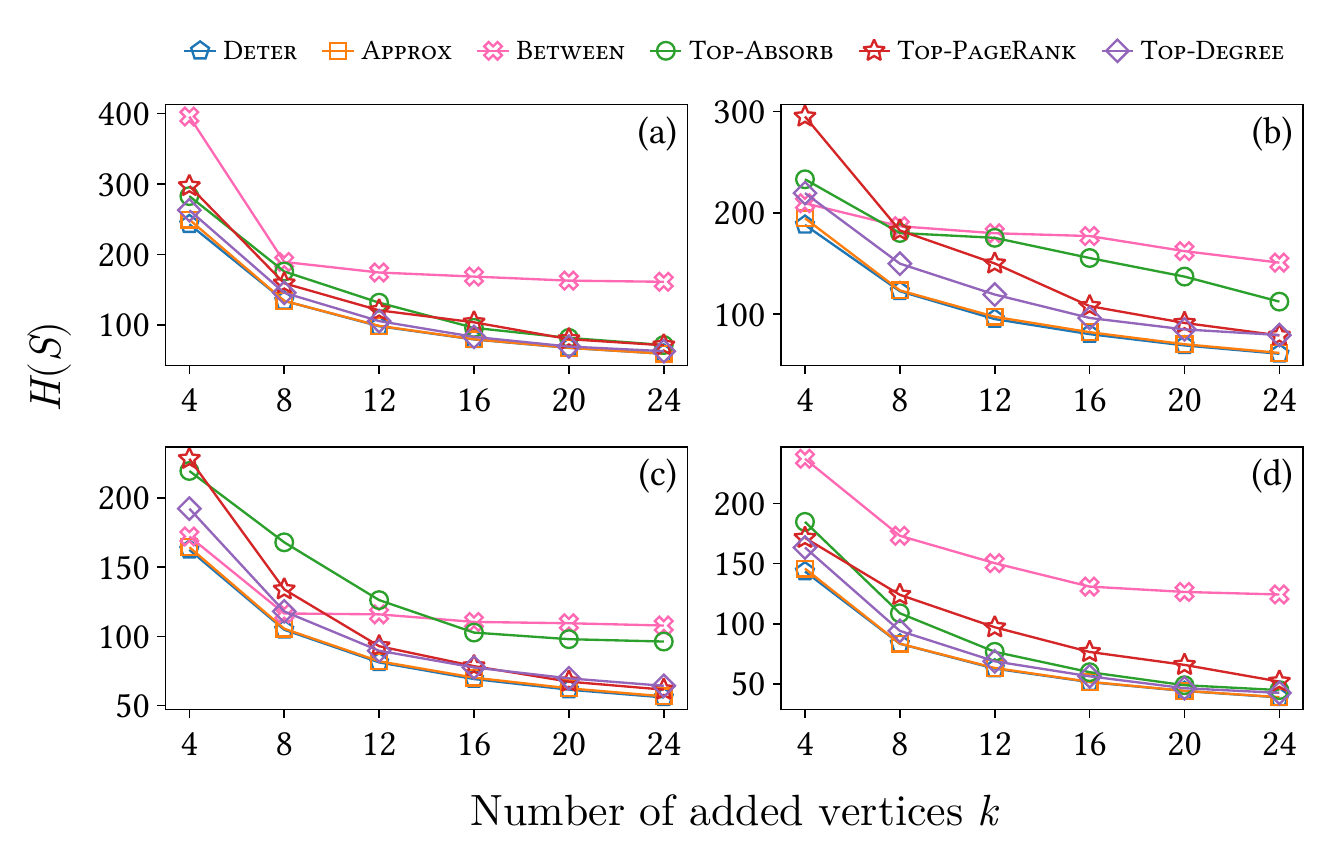}
    \caption{GWC \(\manc{S}\) of vertex group \(S\) computed by six different algorithms, \textsc{DeterMinGWC} (\textsc{Deter}), \textsc{ApproxMinGWC} (\textsc{Approx}), \textsc{Between}, \textsc{Top-Absorb},  \textsc{Top-Degree}, and \textsc{Top-PageRank} on four networks: Astro-ph (a), Pretty Good Privacy (b), Sister cities (c), and CA-HepPh (d).}
    \Description{GWC of vertex group computed by six different algorithms on four networks: Astro-ph (a), Pretty Good Privacy (b), Sister cities (c), and CA-HepPh (d).}
    \label{pic:medium-effect-line}
\end{figure}

\textbf{Comparison of Efficiency between \textsc{DeterMinGWC} and \textsc{ApproxMinGWC}.}
Although both \textsc{DeterMinGWC} and \textsc{ApproxMinGWC} are very effective, below we demonstrate that \textsc{ApproxMinGWC} is significantly more efficient than \textsc{DeterMinGWC} and the baseline \textsc{Between}, particularly in large networks. To this end, we compare the running time of \textsc{DeterMinGWC}, \textsc{Between}, and \textsc{ApproxMinGWC}. We also demonstrate the GWC of vertex group returned by \textsc{DeterMinGWC} and \textsc{ApproxMinGWC} on a set of real-life networks. For each network, we use \textsc{DeterMinGWC}, \textsc{Between}, and \textsc{ApproxMinGWC} to select \(k=10\) vertices. We present the results in \tabref{tab:running-time}, from which we observe that the running time of \textsc{ApproxMinGWC} is much lower than that of \textsc{DeterMinGWC} and \textsc{Between}, while the vertex sets returned by \textsc{DeterMinGWC} and \textsc{ApproxMinGWC} have almost the same GWC.

The efficiency advantage of \textsc{ApproxMinGWC} over \textsc{DeterMinGWC} becomes more obvious for larger networks. \tabref{tab:running-time} shows the difference of running time among \textsc{DeterMinGWC}, \textsc{Between}, and \textsc{ApproxMinGWC} increases rapidly when the networks become larger. For those networks marked with \(\ast\) in \tabref{tab:running-time},  \textsc{DeterMinGWC} cannot run due to the memory and time limitations, while \textsc{ApproxMinGWC} still works well. In particular, \textsc{ApproxMinGWC} scales to massive networks with one million nodes. For example, for the roadNet-CA network with almost two million nodes, \textsc{ApproxMinGWC} solves the MinGWC problem in about four hours. Thus, \textsc{ApproxMinGWC} is both efficient and effective, being scalable to massive networks.

\textbf{Impact of Varying Error Parameter for \textsc{ApproxMinGWC}.}
Numerical results reveal that the error parameter \(\epsilon\) plays a crucial role in determining the performance of our algorithm. Therefore, we further examine the influence of varying error parameter on the accuracy and efficiency of \textsc{ApproxMinGWC}. Specifically, we conduct a series of experiments across several real-world networks, adjusting \(\epsilon\) within the range of 0.15 to 0.4. According to the results shown in \tabref{tab:running-time}, we observe that the running time of \textsc{ApproxMinGWC} increases in proportion to \(\epsilon^{-2}\), aligning with its time complexity of \(\tilde{O}(km\epsilon^{-2})\). In addition, \tabref{tab:running-time} reveals that the relative error of \textsc{ApproxMinGWC} remains considerably small as \(\epsilon\) changes. We also observe that the relative error is also sensitive to the change of \(\epsilon\), which indicates that the accuracy of our algorithm is effectively controlled by the selection of error parameter.

\begin{table}[htbp]
    \tabcolsep=3pt
    \fontsize{7.5}{8.3}\selectfont
    \caption{The running time (seconds, \(s\)) and the relative error (\(\times10^{-2}\)) of \textsc{DeterMinGWC}(\textsc{Deter}), \textsc{Between}, and \textsc{ApproxMinGWC}(\textsc{Approx}) on several realistic networks for various \(\epsilon\)}
    \label{tab:running-time}
    \begin{tabular}{@{}ccccccccccccccc@{}}
        \toprule
        \multirow{3}{*}{Network}
                                          &
        \multicolumn{8}{c}{Running time (s)}
                                          &
        \multicolumn{6}{c}{Relative error (\(\times 10^{-2}\))}                                                                                           \\ \cmidrule(l){2-15}
                                          &
        \multirow{2}{*}{\textsc{Deter}}   &
        \multirow{2}{*}{\textsc{Between}} &
        \multicolumn{6}{c}{\textsc{Approx} with various \(\epsilon\)}
                                          &
        \multicolumn{6}{c}{\textsc{Approx} with various \(\epsilon\)}                                                                                     \\ \cmidrule(l){4-15}
                                          &
                                          &
                                          &
        \(0.4\)                           &
        \(0.35\)                          &
        \(0.3\)                           &
        \(0.25\)                          &
        \(0.2\)                           &
        \(0.15\)                          &
        \(0.4\)                           &
        \(0.35\)                          &
        \(0.3\)                           &
        \(0.25\)                          &
        \(0.2\)                           &
        \(0.15\)                                                                                                                                          \\
        \midrule
        Jazz musicians                    & 0.214 & 0.035 & 0.145 & 0.102 & 0.157 & 0.178 & 0.194 & 0.245 & 1.711 & 1.062 & 0.959 & 0.631 & 0.473 & 0.323 \\
        Euroroads                         & 0.498 & 1.226 & 0.178 & 0.191 & 0.239 & 0.277 & 0.281 & 0.504 & 3.796 & 2.652 & 2.457 & 2.154 & 0.261 & 0.147 \\
        Hamster friends                   & 1.183 & 5.860 & 0.272 & 0.311 & 0.348 & 0.450 & 0.629 & 1.147 & 1.169 & 0.658 & 0.266 & 0.250 & 0.132 & 0.081 \\
        Hamster full                      & 1.267 & 7.456 & 0.369 & 0.423 & 0.402 & 0.525 & 1.049 & 1.309 & 2.038 & 1.017 & 0.415 & 0.094 & 0.079 & 0.048 \\
        Facebook (NIPS)                   & 10.20 & 56.63 & 1.726 & 1.973 & 1.956 & 2.919 & 4.349 & 4.957 & 1.919 & 1.343 & 1.050 & 0.649 & 0.072 & 0.036 \\
        CA-GrQc                           & 10.74 & 52.93 & 0.439 & 0.511 & 0.686 & 0.896 & 0.968 & 1.670 & 3.554 & 2.690 & 1.339 & 1.067 & 0.812 & 0.745 \\
        US power grid                     & 17.62 & 102.3 & 0.416 & 0.412 & 0.500 & 0.634 & 0.977 & 1.523 & 3.661 & 2.521 & 1.907 & 1.121 & 0.534 & 0.211 \\
        Reactome                          & 30.13 & 177.9 & 2.676 & 2.981 & 3.593 & 4.680 & 7.329 & 10.71 & 2.255 & 1.850 & 0.895 & 0.668 & 0.479 & 0.178 \\
        Route views                       & 37.55 & 244.7 & 0.495 & 0.521 & 0.582 & 0.704 & 1.024 & 1.695 & 0.705 & 0.381 & 0.326 & 0.184 & 0.031 & 0.020 \\
        CA-HepTh                          & 85.83 & 738.2 & 0.816 & 1.053 & 1.108 & 1.602 & 2.476 & 4.032 & 0.986 & 0.885 & 0.475 & 0.090 & 0.067 & 0.060 \\
        Sister cities                     & 143.3 & 1337  & 2.190 & 2.591 & 3.246 & 3.974 & 7.013 & 13.01 & 0.726 & 0.456 & 0.355 & 0.152 & 0.131 & 0.067 \\
        Pretty Good Privacy               & 157.7 & --    & 2.545 & 3.049 & 3.844 & 5.170 & 7.657 & 12.02 & 0.820 & 0.386 & 0.286 & 0.118 & 0.069 & 0.023 \\
        CA-HepPh                          & 180.8 & --    & 5.200 & 6.008 & 6.601 & 9.775 & 15.74 & 23.20 & 0.765 & 0.662 & 0.321 & 0.207 & 0.053 & 0.033 \\
        Astro-ph                          & 715.8 & --    & 9.719 & 11.57 & 13.51 & 18.26 & 31.31 & 49.04 & 2.003 & 1.490 & 1.041 & 0.208 & 0.042 & 0.026 \\
        CAIDA                             & 2272  & --    & 4.255 & 5.088 & 6.355 & 7.838 & 11.28 & 20.20 & 1.148 & 0.515 & 0.228 & 0.216 & 0.141 & 0.122 \\
        Brightkite*                       & --    & --    & 15.81 & 19.57 & 23.31 & 32.68 & 54.21 & 84.88 & --    & --    & --    & --    & --    & --    \\
        Livemocha*                        & --    & --    & 129.9 & 150.5 & 177.0 & 250.6 & 392.1 & 639.5 & --    & --    & --    & --    & --    & --    \\
        WordNet*                          & --    & --    & 61.55 & 74.67 & 83.89 & 124.5 & 221.4 & 360.3 & --    & --    & --    & --    & --    & --    \\
        Gowalla*                          & --    & --    & 98.96 & 114.4 & 129.1 & 195.7 & 342.8 & 587.9 & --    & --    & --    & --    & --    & --    \\
        com-DBLP*                         & --    & --    & 155.9 & 204.0 & 232.7 & 399.1 & 756.8 & 1106  & --    & --    & --    & --    & --    & --    \\
        Amazon*                           & --    & --    & 249.4 & 337.5 & 349.5 & 652.3 & 1082  & 1644  & --    & --    & --    & --    & --    & --    \\
        roadNet-PA*                       & --    & --    & 1604  & 1932  & 2215  & 3349  & 5714  & 8551  & --    & --    & --    & --    & --    & --    \\
        YouTube*                          & --    & --    & 1018  & 1228  & 1482  & 2197  & 3661  & 5381  & --    & --    & --    & --    & --    & --    \\
        roadNet-TX*                       & --    & --    & 2225  & 2802  & 3135  & 4464  & 9338  & 12464 & --    & --    & --    & --    & --    & --    \\
        roadNet-CA*                       & --    & --    & 3664  & 4637  & 5445  & 8214  & 14458 & 21565 & --    & --    & --    & --    & --    & --    \\
        \bottomrule
    \end{tabular}
\end{table}

\section{Related Work}

As a fundamental concept in social network analysis and complex networks, network centrality has received considerable attention from the scientific community~\cite{Ne10}. The value of centrality metrics can be used to rank nodes or edges in networks~\cite{LuChReZhZhZh16}, with an aim to identify important nodes or edges for different applications. In the past decades, a host of centrality measures were presented to describe and characterize the roles of nodes or edges in networks~\cite{WhSm03,BoVi14,BeKl15,BoDeRi16,LiZh18,YiShLiZh18,BaXuZh22,BaZh22}. Among various centrality metrics, a class of popular ones are based on shortest paths, of which   betweenness centrality~\cite{Br01,GeSaSc08} and closeness centrality~\cite{Ba48,Ba50} are probably the two most popular ones. A main drawback of these two measures is that they only consider the shortest paths, neglecting contributions from other paths, even the second shortest paths. To overcome this flaw, betweenness centrality based on random walks~\cite{Ne05} and Markov centrality~\cite{WhSm03} were introduced, both of which display better discriminate power than their traditional shortest path based centrality metrics~\cite{Ne05,BeWeLuMe16}, since they incorporate the contributions from all paths between different node pairs.

In addition to the measures of centrality for an individual vertex, many metrics for centrality of node group were presented, in order to measure the importance of a node group. Most previous  node group centrality measures are based on graph structure~\cite{XuReLiYuLi17} or dynamic processes~\cite{MaMaGi15}. For example, based on the shortest paths, group betweenness~\cite{DoelPuZi09,Yo14,MaTsUp16} and group closeness~\cite{EvBo99,ZhLuToGu14,BeGoMe18} were proposed. Similar to the case of individual nodes, group betweenness and group closeness ignore contributions from non-shortest paths. To capture the contributions of all paths, motivated by the work of White and Smyth~\cite{WhSm03}, absorbing random-walk group centrality was proposed~\cite{MaMaGi15}. An algorithm for minimizing the absorbing random-walk group centrality with a cardinality constraint was also proposed~\cite{MaMaGi15}. It has a $O(kn^3)$, and is thus not applicable to large networks.
Our group walk centrality is different from that in~\cite{MaMaGi15}. Moreover, our fast algorithm \textsc{ApproxMinGWC} has a nearly-linear time complexity.

\section{Conclusion}

The hitting time of random walks arises in many practical scenarios. However, the time cost of exactly computing hitting time is prohibitively expensive. In this paper, we studied a walk centrality $H_j$ for a single vertex and Kemeny constant $\kem$ of a graph \(\gr=(V,E,w)\) with $n$ vertices and $m$ edges, both of which are actually weighted averages of hitting times and have found wide applications. We established a link between the two quantities, and reformulated them in terms of quadratic forms of the pseudoinverse of the Laplacian $\lap$ for graph $\gr$. Moreover, we provided a randomized approximation algorithm \textsc{ApproxHK} with probabilistic guarantee, which computes the walk centrality $H_j$ for all vertices $j \in V$ and Kemeny constant $\kem$ in nearly linear time with respect to the edge number $m$.

We then introduced the group walk centrality for a group $H(S)$ of vertices, which extends the notion of walk centrality of an individual vertex to multiple vertices in set $S$. We proved that as a set function, $H(\cdot)$ is monotone and supermodular.
We established a link among $H(S)$, the Kemeny constant, and group random detour time, which deepens the understanding of the cruciality of node group \(S\). Analogous to the property of the Kemeny constant, this connection is very interesting, whose interpretation presents an opportunity for future research. We then considered the NP-hard optimization problem MinGWC of finding the set $S^*$ of $k \ll n$ vertices, aiming to minimize $H(S^*)$. We then developed two greedy approximation algorithms, \textsc{DeterMinGWC} and \textsc{ApproxMinGWC}, which solve the problem MinGWC by iteratively selected $k$ vertices. The former obtains a $(1-\frac{k}{k-1} \frac{1}{e})$ approximation factor in time $O(n^3)$; while the latter achieves a $(1-\frac{k}{k-1}\frac{1}{e}-\epsilon)$ approximation ratio in time $\Otil (km\epsilon^{-2})$ for any $0<\epsilon <1$.

Finally, we conducted extensive experiments on various real and model networks, in order to verify the performance of our proposed algorithms. The results show that algorithm \textsc{ApproxHK} is both efficient and accurate for approximating walk centrality of all individual vertices and the Kemeny constant $\kem$. For the two greedy algorithms \textsc{DeterMinGWC} and \textsc{ApproxMinGWC}, both of them are very accurate, often returning almost optimal solutions. Particularly, the randomized algorithm \textsc{ApproxMinGWC} is scalable to large networks, which quickly returns good approximation solutions in networks with over a million vertices.

\section*{Acknowledgements}

This work was supported by the National Natural Science Foundation of China (Nos. 62372112, U20B2051, and 61872093).

\bibliographystyle{ACM-Reference-Format}

\end{document}